\documentclass[10pt,journal,compsoc]{IEEEtran}
\usepackage{amsmath,amsfonts}
\usepackage{array}
\usepackage[caption=false,font=normalsize,labelfont=sf,textfont=sf]{subfig}
\usepackage{textcomp}
\usepackage{stfloats}
\usepackage{url}
\usepackage{verbatim}
\usepackage{graphicx}
\usepackage{cite}
\usepackage[hidelinks]{hyperref}
\usepackage{lipsum} 
\usepackage{textcomp}
\usepackage{listings}
\usepackage[center]{caption}
\usepackage{gensymb}
\usepackage[utf8]{inputenc}
\usepackage{cite}
\usepackage{fancyhdr}
\usepackage{amsmath,amssymb,amsfonts}
\usepackage{algorithm}
\usepackage{algpseudocode}
\usepackage{graphicx}
\usepackage{textcomp}
\usepackage{xcolor}
\usepackage{amsmath}
\usepackage{pifont}
\usepackage{soul}
\newcommand{\cmark}{\ding{51}}%
\newcommand{\xmark}{\ding{55}}%
\definecolor{cb1}{HTML}{F9CECC}
\definecolor{cb2}{HTML}{4877BD}

\definecolor{axi1}{HTML}{1F77B4}
\definecolor{axi2}{HTML}{FF7F0E}
\definecolor{axi3}{HTML}{2CA02C}
\definecolor{axi4}{HTML}{D62728}

\definecolor{cb3}{HTML}{80ABB3}
\definecolor{cb2d}{HTML}{445A5E}

\definecolor{clg}{HTML}{C2C2C2}
\definecolor{cdg}{HTML}{5E5E5E}
 
\usepackage{pgfplots}
\usepgfplotslibrary{external}
\pgfplotsset{compat=newest}
\usepgfplotslibrary{statistics}
\usepgfplotslibrary{groupplots}
\usepackage{tikz}
\usetikzlibrary{patterns}

\usepackage{graphicx}

\usepackage{wrapfig}
\usepackage{multirow}
\usepackage{microtype}
\usepackage{mathtools}
\usepackage[normalem]{ulem}

\newcommand\review[1]{\textcolor{black}{#1}}

\usepackage{amsthm}
\newtheorem{lemma}{Lemma}%{Lemma}we
\newtheorem{theorem}{Theorem}

\def\BibTeX{{\rm B\kern-.05em{\sc i\kern-.025em b}\kern-.08em
    T\kern-.1667em\lower.7ex\hbox{E}\kern-.125emX}}

\title{TOP: Towards Open \& Predictable Heterogeneous SoCs}
\author{Luca Valente, 
Francesco Restuccia,
Davide Rossi,~\IEEEmembership{Member,~IEEE} \\
Ryan Kastner,~\IEEEmembership{Fellow,~IEEE} 
Luca Benini, ~\IEEEmembership{Fellow,~IEEE} 
\IEEEcompsocitemizethanks{ \IEEEcompsocthanksitem Luca Valente, Luca Benini, and Davide Rossi are with the Department of Electrical, Electronic and Information Engineering, University of Bologna, 40136 Bologna, Italy.  Luca Benini is also with the Integrated Systems Laboratory (IIS), ETH Zürich, 8092 Zürich, Switzerland.
\IEEEcompsocthanksitem Francesco Restuccia and Ryan Kastner are with the Computer Science and Engineering, University
of California at San Diego, San Diego, CA 92093 USA.}
\thanks{This work was supported by Technology Innovation Institute, Secure Systems Research Center, Abu Dhabi, UAE, PO Box: 9639, by the Spoke 1 on Future HPC of the Italian Research Center on High-Performance Computing, Big Data and Quantum Computing (ICSC) funded by MUR Mission 4 - Next Generation EU, and by the European Project EuroHPC JU The European Pilot (g.a. 101034126), and by KDT TRISTAN project (g.a.101095947).}}

\begin{document}

\bstctlcite{IEEEexample:BSTcontrol}

\IEEEtitleabstractindextext{%
\begin{abstract}
Ensuring predictability in modern real-time Systems-on-Chip (SoCs) is an increasingly critical concern for many application domains such as automotive, robotics, and industrial automation. An effective approach involves the modeling and development of hardware components, such as interconnects and shared memory resources, to evaluate or enforce their deterministic behavior. Unfortunately, these IPs are often closed-source, and these studies are limited to the single modules that must later be integrated with third-party IPs in more complex SoCs, hindering the precision and scope of modeling and compromising the overall predictability. With the coming-of-age of open-source instruction set architectures (RISC-V) and hardware, major opportunities for changing this status quo are emerging. This study introduces an innovative methodology for modeling and analyzing State-of-the-Art (SoA) open-source SoCs for low-power cyber-physical systems. Our approach models and analyzes the entire set of open-source IPs within these SoCs and then provides a comprehensive analysis of the entire architecture.
We validate this methodology on a sample heterogenous low-power RISC-V architecture through RTL simulation and FPGA implementation, minimizing pessimism in bounding the service time of transactions crossing the architecture between 28\% and 1\%, which is considerably lower when compared to similar SoA works.
\end{abstract}
}

\maketitle

\section{Introduction}\label{sec:intro}

The exponential growth of cyber-physical systems (CPS) (e.g., self-driving cars, autonomous robots, ...) and related applications has been fueled by the increase in computational capabilities of heterogeneous low-power Systems-on-Chip (SoCs).
These SoCs are complex computing platforms composed of a set of different hardware computing units (e.g., CPUs, hardware accelerators), each tailored to a specific target application, sharing a set of resources (memory, sensors) through interconnects\cite{jiang2022axi,biondi2021sphere,jiang2022bluescale,restuccia2020axi,ooocoresareSOnasty}.
While integrating multiple computing units on the same platform has enabled efficient scale-up of computational capabilities, it also poses significant challenges when it comes to assessing their \textit{timing predictability}, which is a requirement for CPSs dealing with real-time and safety-critical applications: the primary challenge arises from resource contentions that emerge when multiple active agents within the SoC must access the same shared resources\cite{jiang2022axi,restuccia2020axi,biondi2021sphere,fernandez2014contention,enabling_compositionability,jiang2022bluescale,ooocoresareSOnasty}.

Numerous research efforts have focused on enhancing the \review{timing} predictability of heterogeneous Systems-on-Chip (SoCs). \review{This includes safely upper bounding execution times for data transfers \cite{hassanzoni,restuccia2022bounding,wu2023ditty} or the deadline miss ratio for critical tasks\cite{jiang2022axi,biondi2021sphere,jiang2022bluescale}, with the smallest possible pessimism}. 
These efforts have predominantly focused on modeling and analyzing commercial DDR protocols \cite{hassanzoni}, memory IPs \cite{hassan}, and memory controllers \cite{mirosanlou2020mcsim}, but also predictable interconnects\cite{restuccia2020axi,jiang2022axi} and on-chip communication protocols\cite{MPAM-analysis}.
\review{Regrettably, despite their value, these studies are scattered, with each one focusing on only one of these resources at a time, resulting in being overly pessimistic\cite{ooocoresareSOnasty}}.

\review{Modeling and analysis of communication protocols are done speculatively on abstract models, thus reducing their real-world applicability.}
Recent works for modeling and analysis of IPs (memories, memory controllers, interconnect, etc.) have to address the unavailability of cycle-accurate RTL descriptions. Many of these IPs are either entirely closed-source\cite{hassanzoni} or provide loosely-timed behavioral models\cite{mirosanlou2020mcsim,ooocoresareSOnasty} or just $\mu$architectural descriptions\cite{restuccia2020axi,jiang2022axi,jiang2022bluescale}.
\review{In essence, the fragmented and proprietary nature of commercial and research IPs restricts studies to the particular IP, greatly reducing the accuracy achievable through system-level analysis.}
For example, Restuccia et al. in \cite{restuccia2022bounding} bound the access times of multiple initiators on FPGA reading and writing from/to the shared DDR memory.
The proposed upper bounds' pessimism is between 50\% and 90\%: even though they finely modeled and analyzed the proprietary interconnect, the authors did not have access to its RTL nor to the memory controller and IP.
The same applies to Ditty \cite{wu2023ditty}, which is a predictable cache coherence mechanism. 
In Ditty, even though the caches' timing is finely modeled, the overall execution time can be up to 3$\times$ bigger than the theoretical upper bounds, as the authors did not model other components. 
Another example is $\textrm{AXI-IC}^{\text{RT}}$ \cite{jiang2022axi}, an advanced AXI interconnect with a sophisticated scheduler which allows transaction prioritization based on importance.
While proposing a highly advanced interconnect with a tightly coupled model, the authors do not extend the model to the other components of the SoC, even when assessing the deadline miss ratio and benchmarking the architecture.

The emergence of open-source hardware creates a major opportunity for building accurate end-to-end models for real-time analysis of cutting-edge heterogeneous low-power SoCs \cite{open-source-hw,cheshire,shaheen}: the openness of the IPs allows for cycle-accurate analysis of the whole architecture from the interconnects to the shared resources.
Yet, investigations and successful demonstrations in this direction are still scarce, primarily because open hardware has only very recently reached the maturity and completeness levels required to build full heterogeneous SoCs\cite{stinky}.
\review{In this context, this is the first work to bridge the gap between open-source hardware and timing analysis, demonstrating a methodology that successfully exploits the availability of the source code  to provide fine-grained upper bounds of the system-level data transfers.}
We leverage a set of open-source IPs from the PULP family, one of the most popular open-hardware platforms proposed by the research community~\cite{open-source-hw,pulp-platform-git}. 

\begin{figure}
    \centering
    \includegraphics[width=\linewidth]{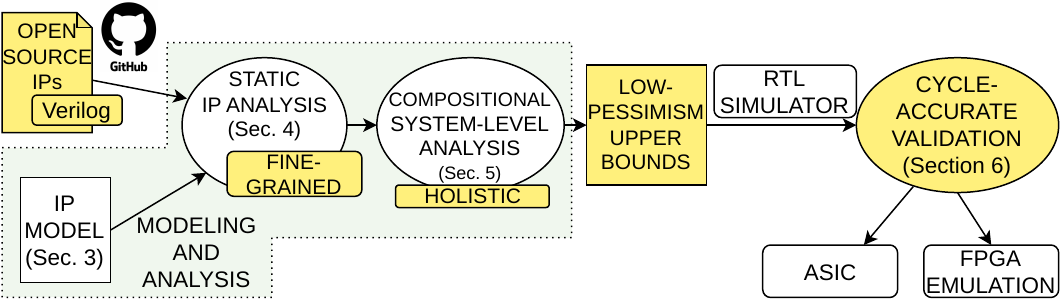}
    \caption{Proposed methodology.}
    \label{fig:methodology}
\end{figure}

\review{Figure \ref{fig:methodology} shows the proposed methodology, highlighting the novel contributions in yellow. It consists of (i) a model for standalone IPs composing modern heterogeneous low-power SoCs, (ii) a static analysis of the RTL code of such components, and (iii) a compositional mathematical analysis of the whole system to upper bound the response time of the interactions between managers (initiators) and shared subordinates (targets), considering the maximum interference generated by the interfering managers.}
\review{Figure \ref{fig:methodology} highlights the differences between the proposed methodology and previous studies also based on a static and compositional approach\cite{enabling_compositionability,restuccia2022bounding,ooocoresareSOnasty}.}
\review{While previous works typically focus on one IP at a time \cite{restuccia2022bounding}, or rely on loosely-timed models\cite{ooocoresareSOnasty}, thereby limiting the overall accuracy, our approach is the first to model and analyze all the IPs directly from the RTL source code to build a holistic system-level analysis.}
\review{This limits the proposed upper bounds' pessimism between 28\% and just 1\%, in isolation and under interference, which is considerably lower when compared to similar SoA works for closed-source or loosely-timed platforms\cite{biondi2021sphere,wu2023ditty,jiang2022axi,jiang2022bluescale,restuccia2020axi,hassanzoni}, as better detailed in Section \ref{sec:related}.}
\review{We demonstrate our methodology on a completely open-source prototype of a heterogeneous low-power open-source SoC} for embedded systems composed of a Linux-capable host core, a parallel accelerator, a set of IOs, and on-chip and off-chip memories.

The manuscript is organized as follows: Section \ref{sec:arch} presents the target open-source RISC-V-based SoC architecture, and Section \ref{sec:model} discusses the model we apply to its different components. Section \ref{sec:wcea} analyzes the components to specialize the generic model to each of them, and Section \ref{sec:system-level-analysis} provides the system-level analysis of the architecture. Finally, Section \ref{sec:exps} validates the results with cycle-accurate experiments (on simulation and FPGA), Section \ref{sec:related} compares this work with the SoA. Section \ref{sec:conclusion} concludes the manuscript.

\section{Architecture}\label{sec:arch}

\begin{figure}[t]
    \centering
    \includegraphics[width=0.45\textwidth]{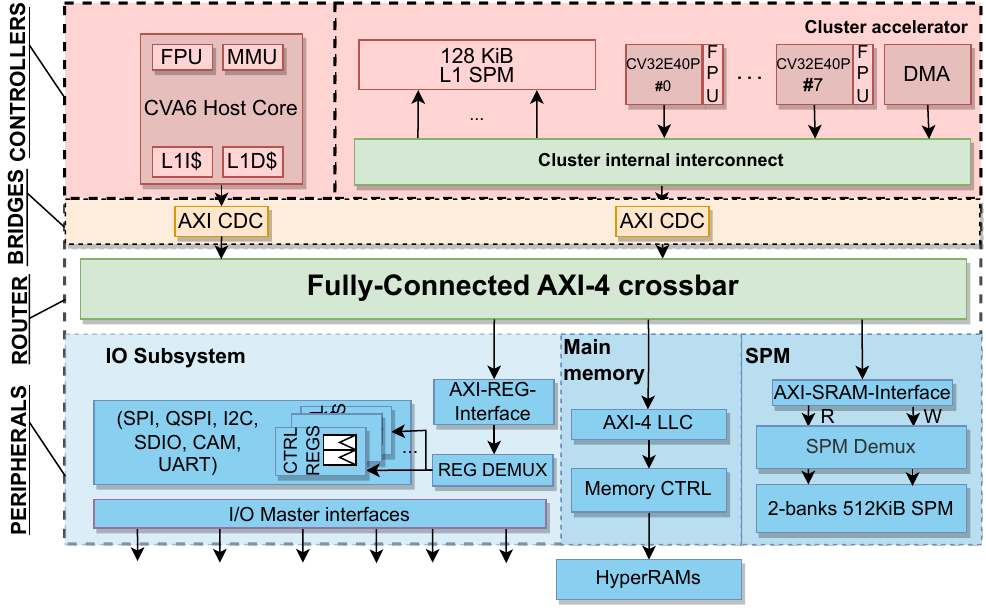}
    \caption{Sample architecture.}
    \label{fig:he-arch}
\end{figure}

Fig.~\ref{fig:he-arch} shows the architectural template we target. 
It also shows the four classes of hardware modules we identify in the architecture under analysis, namely (i) \emph{controllers}, (ii) the main \emph{crossbar}, (iii) \emph{bridges}, and (iv) \emph{peripherals}, which we model in the next Section.
The architecture leverages a set of fully open-source PULP IPs\cite{pulp-platform-git}.
It is based on Cheshire\cite{cheshire}, an open-source host platform consisting of an RV64 Linux-capable CPU, a set of commodity IOs (SPI, SDIO, UART, ...), and an AXI-based crossbar with a configurable number of subordinate and manager ports for easy integration of accelerators and resources.
Our platform includes a parallel accelerator and a low-power lightweight HyperBUS memory controller\cite{hulk-v}, connected to the crossbar.

The host CPU is CVA6\cite{cva6-git}, which is a six stages, single-issue, in-order, 64-bit Linux-capable RISC-V core, supporting the RV64GC ISA variant, SV39 virtual memory with a dedicated Memory Management Unit (MMU), three levels of privilege (Machine, Supervisor, User), and PMP~\cite{pie}.
\review{CVA6 features private L1 instruction and caches, operating in parallel, with the latter being able to issue multiple transactions.}
\review{When needed, CVA6 can offload computation-intensive tasks to} the parallel hardware accelerator, the so-called PULP cluster\cite{pulp-cluster-git}.
It is built around 8 CV32E4-based cores~\cite{cv32e40p} sharing 16$\times$8 kB SRAM banks, composing a 128 kB L1 \review{Scratchpad Memory} (SPM). 
\review{The cluster features a DMA to perform data transfers between the private L1SPM and the main memory: data movement is performed via software-programmed DMA transfers.}
\review{Once the data are available inside the L1SPM, the accelerator starts the computation.}

CVA6 and the cluster are the managers of the systems connected to the main AXI crossbar\cite{Kurth_2021}, which routes their requests to the desired subordinates according to the memory map.
\review{A manager can access any subordinate in the system.}
The main subordinates of the systems are, respectively, (i) the on-chip SRAM memory, (ii) the IO subsystem, and (iii) the off-chip main memory \review{with a tightly coupled Last Level Cache (LLC)}.
The on-chip memory is used for low-latency, high-bandwidth data storage.
The APB subsystem is used to communicate with off-chip sensors or memories through the commodity IOs.
\review{The off-chip main memory is where the code and the shared data are stored.}
Differently from high-end embedded systems relying on relatively power-hungry and expensive DDR3/4/5 memories, the platform under analysis adopts HyperRAMs as off-chip main memory, which are fully-digital low-power small-area DRAMs with less than 14 IO pins and that provide enough capacity to boot Linux \cite{shaheen} and bandwidth for IoT applications\cite{hyperram_low_pincount,hulk-v}.

\section{Model}\label{sec:model}

This section presents the model we construct for the different components of our SoC.
\review{Our aim is to propose a general model that describes the characteristics of the components and that can be re-targeted to different IPs and novel architectures, regardless of the number of integrated controllers and peripherals}.
This work is also an effort to provide base support to stimulate further studies in predictability improvements and analysis for open hardware architectures.

\subsection{Communication model}

We identify four classes of hardware modules in the architecture under analysis, shown in Fig. \ref{fig:he-arch}, namely (i) \emph{controllers}, (ii) the main \emph{crossbar}, (iii) \emph{bridges}, and (iv) \emph{peripherals}.
\review{As the AXI standard is the main communication standard used to implement non-coherent on-chip communications\cite{Kurth_2021}, we discuss here its main features.}
It defines a manager-subordinate interface enabling simultaneous, bi-directional data exchange and multiple outstanding transactions. 
Fig. \ref{fig:axi-chan} shows the AXI channel architecture and information flow.
Bus transactions are initiated by a \emph{controller} (exporting a manager interface), submitting a transaction request to read/write data to/from a subordinate interface through AR or AW channels, respectively.
A request describes the starting target address and a \textit{burst length}.
After the request phase, in case of a read, data are transmitted through the R channel.
In case of a write, data are provided by the \emph{controller} to the target \emph{peripheral} through the W channel.
Upon completing a write transaction, the \emph{peripheral} also sends a beat on the B channel to acknowledge the transaction's completion. 
For multiple in-flight write transactions, the standard enforces strict in-order access to the W channel: the data on the W channel must be propagated in the same order as the AW channel requests.
Even though the standard does not require it, many commercial and open-source platforms apply the same policy for reads, typically to limit the system's overall complexity, as reported in their documentation~\cite{zynq-700,axi-bram}.

\subsection{Controller model}\label{model:controller}

\emph{Controllers} have an active role on the bus. Each \emph{controller} exports an AXI manager interface, through which it initiates requests for bus transactions directed to the \emph{peripherals}.
A generic \emph{controller} $C_i$ can be described through two parameters: the maximum number of outstanding read/write transactions that it can issue in parallel, denoted with $\phi_{R/W}^{C_i}$, and their relative burst length $\beta_i$.
While our model and analysis can be applied to a generic architecture, the system under analysis features as \emph{controllers} a CVA6 core~\cite{cva6-git} and a cluster accelerator~\cite{pulp-cluster-git} (see Section~\ref{sec:arch}).
Bus transactions issued by the cluster interfere with those issued by CVA6 and vice-versa.
CVA6 is assumed to compute a critical periodic workload, running on top of a Real-time Operating System (RTOS).
\review{The PULP cluster executes computation-intensive tasks and issues bus transactions through its DMA.}
\review{Contention internal to the PULP cluster has been profiled in detail in \cite{cluster-contention}.}
\review{However, our analysis provides the worst-case data transfer time in accessing the shared \emph{peripherals} to support the safe scheduling and execution of critical tasks within their deadline.}
\review{We specifically focus on interference in accessing the shared resources. Modeling the internal effects of \emph{controllers}, such as pipeline stalls in the core or contention within the accelerator, is beyond the scope of this work.}

\begin{figure}[t]
    \centering
    \includegraphics[width=0.9\linewidth]{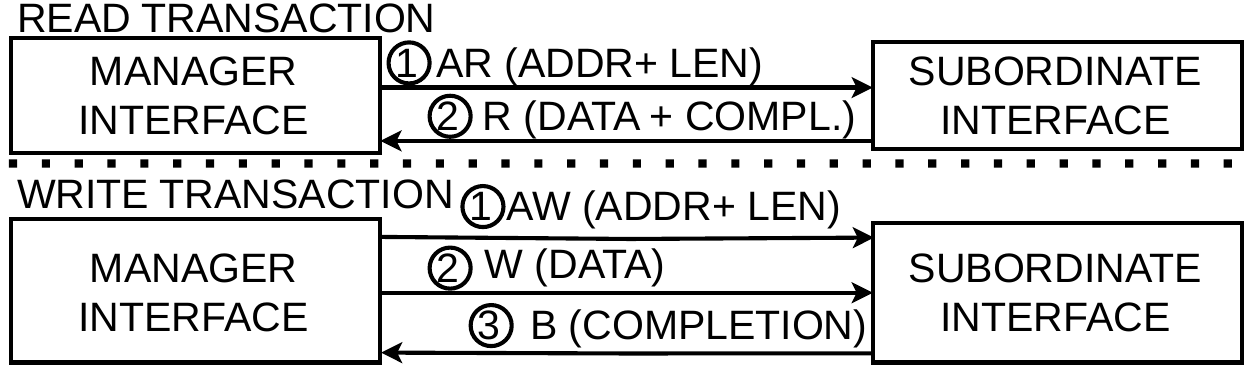}
    \caption{AXI Channel architecture}
    \label{fig:axi-chan}
%    \vspace{-3mm}
\end{figure}

\subsection{Peripheral model}\label{model:peripheral}

\emph{Peripherals} export a \emph{subordinate} interface through which they receive and serve the bus transactions.
\review{The \emph{peripherals} deployed in the system are heterogeneous. Nonetheless, our model offers a set of parameters representative of a generic peripheral, and it is not tied to a specific communication protocol.}
\review{It} works as the baseline for the analysis of any \emph{peripheral} deployed in the system under analysis.
The generic \emph{peripheral} $P_j$ is characterized with two sets of parameters: (i) the maximum number of supported outstanding reads ($\chi_{R}^{P_j}$) and write ($\chi_{W}^{P_j}$) transactions; (ii) the maximum number of cycles incurred from the reception of the request to its completion, for a read ($d_{R}^{P_j}$) and a write ($d_{W}^{P_j}$) transaction in isolation.
$d_{R}^{P_j}$ and $d_{W}^{P_j}$ are composed of two contributions: (i) the \emph{data time}, defined as the time required for the \emph{peripheral} to send or receive one word of data ($t_{\text{DATA}}$) multiplied by the burst length of the transaction in service ($\beta_i$) and (ii) the \emph{control overhead} $t_{\text{CTRL}}$, defined as the maximum time elapsing between accepting the request and the availability of the first word of data (reads) or availability to receive data (writes).
From the previous considerations, $d_{R/W}^{P_j} = t_{\text{CTRL}}^{P_j} + t_{\text{DATA}}^{P_j}\cdot \beta$.
We define two extra parameters $\rho^{P_j}$ and $\theta^{P_j}$.
The first indicates the level of pipelining in serving multiple transactions. 
$\rho^{P_j}=1$ means that each stage of $P_j$ does not stall the previous, and transactions are served in a pipelined fashion, while $\rho^{P_j}=0$ indicates that no pipeline is implemented. 
$\theta^{P_j}=0$ indicates that read and write transactions interfere with each other. 
$\theta^{P_j}=1$ indicates that read and write transactions can be handled in parallel by $P_j$.% This is further described in Section~\ref{sec:system-level-analysis}.

\subsection{Main crossbar model}\label{model:router}
We provide here the model of the main \emph{crossbar}, the routing component enabling communication among \emph{controller}s and \emph{peripheral}s. 
Each \emph{controller} has its manager port connected to a subordinate port of the \emph{crossbar}. 
Each \emph{peripheral} has its subordinate port connected to a manager port of the \emph{crossbar}. 
We model the \emph{crossbar} $R_0$ with two sets of parameters: (i) the maximum amount of outstanding read and write transactions that a subordinate port can accept ($\chi_{R}^{R_0}$ and $\chi_{W}^{R_0}$, respectively); and (ii) the maximum overall latency introduced by $R_0$ on each read ($d_{R}^{R_0}$) and write transaction ($d_{W}^{R_0}$).
$d_{R}^{R_0}$ and $d_{W}^{R_0}$ are composed of two contributions: (i) the overall delay introduced by the \emph{crossbar} on a transaction in isolation ($t_{\text{PROP}}$); (ii) the maximum time a request is delayed at the arbitration stage due to the contention generated by interfering transactions ($t_\text{{CON}}^{R_0}$).
From the previous considerations, the propagation latency is modeled as $d_{R/W}^{R_0} = t_{\text{PROP}}^{R_0} + t_{\text{CON}}^{R_0}$.
Such parameters depend on the arbitration policies and routing mechanisms, as we investigate in detail in Section~\ref{sec:wcea}.

\subsection{Bridge model}\label{model:bridge}

Bridges export a single manager interface and a single subordinate interface. They perform protocol/clock conversion between a \emph{controller} and the \emph{crossbar}.
Bridges require a certain number of clock cycles to be crossed but do not limit the number of in-flight transactions and do not create any contention.
We model the bridges with two parameters: the overall maximum delay introduced over a whole transaction for (a) read ($d_{R}^{Q_j}$) and (b) write ($d_{W}^{Q_j}$) transactions.

\section{Analysis of the hardware modules}\label{sec:wcea}

This Section aims to analyze the worst-case behavior of the \emph{peripherals}, \emph{bridges}, and the \emph{crossbar} present in the platform under analysis.
Our approach is compositional -- in this Section, we analyze each hardware component separately, specializing in the generic models introduced in Section~\ref{sec:model}, and bounding the service times at the IP level in isolation.
In the next Section, we provide an overall worst-case analysis at the system level, in isolation and under interference.
We define $t_{\text{CK}}^{P_j}$ as the period period of the clock fed to $P_j$.
%
%To keep our analysis general, we assume that any component under analysis has its own clock so that, for instance, $t_{\text{CK}}^{\text{SPM}}$ is the clock of the SPM memory. 

\subsection{AXI CDC FIFO queues}\label{ssec:axi_cdc_fifo}

AXI CDC FIFOs are leveraged to perform clock-domain crossing between two AXI-based devices. 
The generic AXI CDC FIFO $F_i$ is a \emph{bridge}: we apply here the model presented in Section \ref{model:bridge}. 
It exports a manager interface and a subordinate interface. It is composed of five independent CDC FIFOs, each serving as a buffer for an AXI channel, having depth $D^i_{\text{CDC}}$ (design parameter for the IP under analysis). 

\subsubsection{RTL IP structure} Figure \ref{fig:cdc-bd} shows the block diagram of a CDC FIFO in the platform under analysis.
They are structured following established clock domain crossing (CDC) principles~\cite{Kurth_2021}. 
The design is split into two parts, the transmitter (TX) and the receiver (RX), having different clock domains. 
TX and RX interface through asynchronous signals, namely a counter for data synchronization (synchronized with two-stage Flip-Flops (FFs)) and the payload data signal. 

\subsubsection{Delays analysis} As mentioned earlier, CDC FIFOs are \emph{bridges}: we apply the model presented in Section~\ref{model:bridge}. 
The CDC FIFO under analysis behaves as follows: TX samples the payload data into an FF. In the following cycle, the TX counter is updated.
The TX counter value gets then through two synchronizations FFs -- the updated pointer value is observed by the RX after two clock cycles. 
At that point, RX samples the data in one clock cycle to then propagate it in the following one. 
It follows that crossing the CDC FIFO introduces a fixed delay of one clock cycle of the TX domain ($t_\text{{CK}}^\text{TX}$) and four clock cycles of the RX domain ($t_\text{{CK}}^\text{RX}$). This means that the delay in crossing the CDC FIFO is equal to $t_{\text{CDC}}(t_\text{{CK}}^\text{TX},t_\text{{CK}}^\text{RX}) = t_\text{{CK}}^\text{TX} + 4 \cdot t_\text{{CK}}^\text{RX}$.%We leverage this baseline delay to build the overall latency introduced by $F_i$ on read and write transactions.
We leverage this baseline delay to build the overall latency introduced by $F_i$, interposed between a manager (clocked at $t_\text{{CK}}^{C}$) and a subordinate (clocked at $t_\text{{CK}}^{P}$).

\textit{Read transaction:} 
A read transaction $AR_k$ is composed of two phases: (i) the address propagation phase and (ii) the data phase. 
This means that $F_i$ is crossed twice to complete $AR_k$: during phase (i), the manager is on the TX side, propagating the request. In phase (ii), the subordinate is on the TX side, propagating the data. 
Hence, the propagation latency is $t_{\text{CDC}}(t_\text{{CK}}^{C},t_\text{{CK}}^{P})$ in phase (i) and $t_{\text{CDC}}(t_\text{{CK}}^{P},t_\text{{CK}}^{C})$ in phase (ii).
Adding them together, the propagation latency introduced by $F_i$ on $AR_k$ is equal to:
\begin{equation}\label{eq-delay-CDC}
\footnotesize
       d_{R}^{\text{CDC}}= t_{\text{CDC}}(t_\text{{CK}}^{C},t_\text{{CK}}^{P}) + t_{\text{CDC}}(t_\text{{CK}}^{P},t_\text{{CK}}^{C}) = 5  ( t_{\text{CK}}^{C} + t_{\text{CK}}^{P} )\\
\end{equation}

\textit{Write transaction:} %similar considerations can be applied to write transactions.
A write transaction is composed of three phases: (i) an address phase (manager on the TX side), (ii) a data phase (manager on the TX side), and (iii) a write response phase (subordinate on the TX side). 
Phases (i) and (ii) happen in parallel (see~\cite{ARMAXI} p. 45).
Thus, $t_{\text{CDC}}(t_\text{{CK}}^{C},t_\text{{CK}}^{P})$ is incurred for phases (i) and (ii), and $t_{\text{CDC}}(t_\text{{CK}}^{P},t_\text{{CK}}^{C})$ for phase (iii). 
The delay introduced by $F_i$ on $AW_k$ is equal to the delay introduced in Equation~\ref{eq-delay-CDC}, $d_{W}^{\text{CDC}} = d_{R}^{\text{CDC}}$.

\begin{figure*}[htb]
    \begin{minipage}[t]{.33\textwidth}
        \centering
        \includegraphics[width=\textwidth]{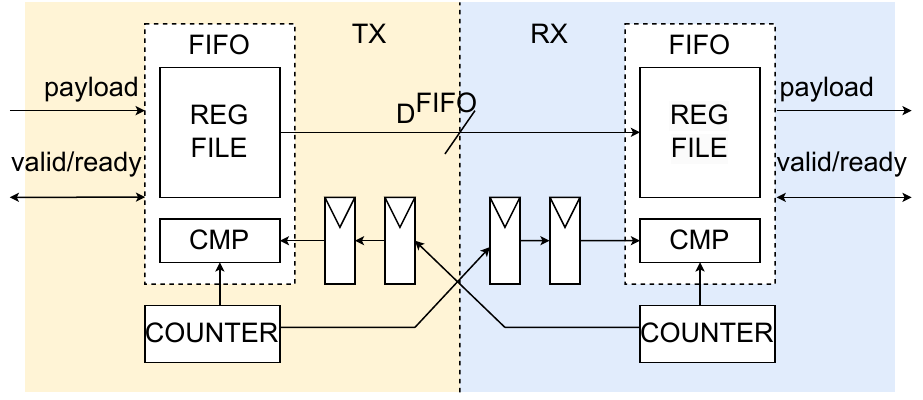}
        \caption{CDC FIFO block diagram.}\label{fig:cdc-bd}
    \end{minipage}
    \hfill
    \begin{minipage}[t]{.33\textwidth}
        \centering
        \includegraphics[width=\textwidth]{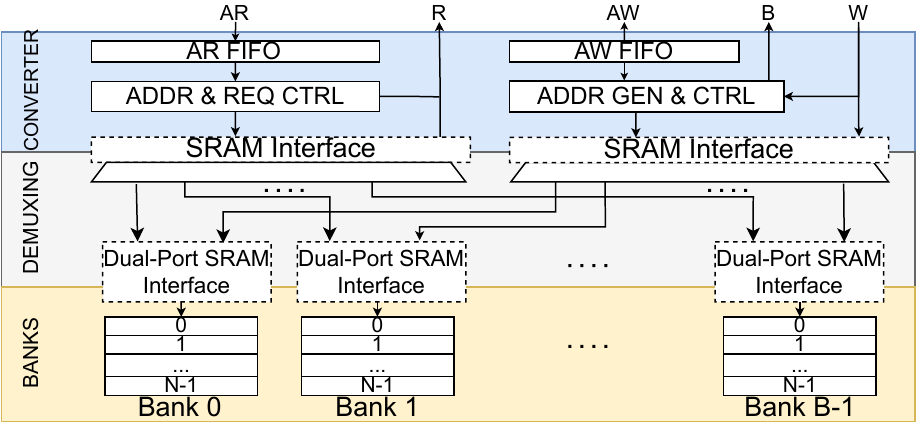}
        \caption{AXI SPM block diagram}\label{fig:axi-spm}
    \end{minipage}
    \hfill
    \begin{minipage}[t]{.32\textwidth}
        \centering
        \includegraphics[width=\textwidth]{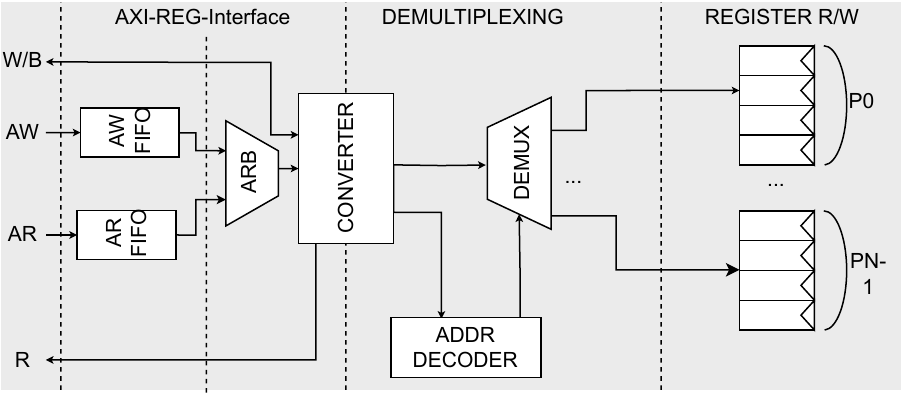}
        \caption{IO subsystem block diagram.}\label{fig:io_subsystem}
    \end{minipage}
\end{figure*}

\subsection{AXI SRAM scratchpad memory (SPM)}\label{ssec:l2spm}

The AXI SPM is a high-speed, low-latency memory component used for temporary data storage -- a block design representation is reported in Figure~\ref{fig:axi-spm}. 
The SPM memory is a \emph{peripheral}: we apply here the model presented in Section~\ref{model:peripheral}. 

\subsubsection{RTL IP structure} 
The first stage of the SPM architecture is represented by a protocol converter (AXI-SRAM-Interface), translating the read and write AXI channels into SRAM-compatible transactions.
Following the converter, an internal demux directs the SRAM transactions to the desired SRAM bank, where the data is stored.
Each SRAM bank provides two independent SRAM ports, one for reads and one for writes, as from the specification of industry-standard SRAM resources~\cite{2port-manual}.

\textit{The AXI-SRAM-Interface} is structured in two submodules, independently managing read and write transactions. 
%
%Each submodule outputs an SRAM-compatible port.
%
The first stage of each submodule is a FIFO queue (of depth $D_{\text{FIFO}}^{\text{SPM}}$) buffering the AXI AW or AR channel, respectively. 
Each submodule features the logic for protocol translation, consisting of (i) saving transaction metadata (starting address and length) and (ii) producing the output SRAM requests. 
For writes, the incoming data on the W channel are directly propagated towards the banks.
The logic operating the protocol conversion generates the address for each W beat. 
For reads, the data coming from the SRAM banks are directly driven on the R channel.
The logic keeps compliance with the AXI standard, adding the last signal or generating write responses when required. 
\textit{The demux} is fully combinatorial and selects the target bank according to the request's address.  
\textit{The SRAM banks} are technology-specific macros instantiated at design time. 
Each SRAM bank's port exports an enable signal, an address signal, and a signal to determine if a transaction is a read or a write.
The SRAM interface expects simultaneous propagation of data and commands for writes; for reads, the data are sent the cycle following the command.

\subsubsection{Delays and parallelism analysis} %At first, we analyze the maximum service delay.

\textit{AXI-SRAM-Interface:} the FIFOs in the converter are only in charge of data buffering -- each FIFO introduces a fixed delay of one clock cycle ($t_{\text{CK}}^{\text{SPM}}$). % (differently from the CDC FIFOs introduced in Section~\ref{ssec:axi_cdc_fifo}).
After the FIFOs, the control logic requires at most one clock cycle ($t_{\text{CK}}^{\text{SPM}}$) to set up the propagation of a burst transaction -- the direct connection over the W and R channels makes the data streaming in a pipeline fashion, adding no further latency.
At the end of a write transaction, the converter takes two clock cycles ($2t_{\text{CK}}^{\text{SPM}}$) to generate the write response: one to acknowledge that the last W beat has been accepted and one to provide the B response.
The same applies to reads, to generate the AXI last signal.
%
%As mentioned earlier, after the protocol conversion, the transactions on the R and W channels are directly merged with SRAM transactions.
%
Summing up the contributions, the control latency introduced by the AXI-SRAM-Interface to each transaction is upper bound by $4t_{\text{CK}}^{\text{SPM}}$ for both reads and writes. 

\textit{Demux:} The demultiplexing is combinatorial: it connects the transaction to the SRAM bank in one clock cycle ($t_{\text{CK}}^{\text{SPM}}$).

\textit{Banks:} As by the definition of the SRAM interface~\cite{2port-manual}, an SRAM bank serves one transaction per clock cycle, which makes $t_{\text{DATA,R/W}}^{\text{SPM}}=t_{\text{CK}}^{\text{SPM}}$. 
For write transactions, the protocol guarantees that the SRAM bank samples the data in parallel with the request (in the same clock cycle). % so there is no control overhead to be added.
For read transactions, the data are served the clock cycle after the bank samples the request. So, it contributes to $t_{\text{CTRL,R}}^{\text{SPM}}$ with one clock cycle ($t_{\text{CK}}^{\text{SPM}}$).
Summing up the contributions, the service time of the SPM in isolation is upper bound by: 
\begin{equation}\label{worst_tcdm_gen_1}
\footnotesize
 t_{\text{CTRL,W}}^{\text{SPM}}  = 5 \cdot t_\text{CK}^{\text{SPM}} ; t_{\text{CTRL,R}}^{\text{SPM}}  = 6 \cdot t_\text{CK}^{\text{SPM}} ; t_{\text{DATA,R/W}}^{\text{SPM}}  =  t_{\text{CK}}^{\text{SPM}}; 
\end{equation}

Consider now the parallelism supported by the SPM. The maximum number of accepted outstanding transactions at the SPM $\chi_R^{\text{SPM}}$ is defined by the depth $D_{\text{FIFO}}^{\text{SPM}}$ of the input buffers implemented in the AXI-SRAM-Interface. Thus,
\begin{equation}
\footnotesize
\chi_R^{\text{SPM}} = \chi_W^{\text{SPM}} = D_{\text{FIFO}}^{\text{SPM}} 
\end{equation}

The \textit{SPM} module under analysis is aggressively pipelined, operations are executed in one clock cycle, and no stall sources are present in the design. Also, as mentioned earlier, read and write transactions do not interfere with each other. From the previous considerations, $\rho^{\text{SPM}} = 1$ and $\theta^{\text{SPM}}=1$. 

\subsection{IO Subsystem}\label{analysis:io}

The IO subsystem is the \emph{peripheral} in charge of writing/reading data to/from the off-chip I/Os. 
We apply here the model presented in Section~\ref{model:peripheral}. 
It is composed of a set of memory-mapped peripheral registers that are accessed through a demux and that manage the datapaths issuing the transactions on the I/O interfaces (e.g., SPI, I2C, etc.).

\subsubsection{RTL IP structure} Figure~\ref{fig:io_subsystem} shows the block diagram of the IO subsystem.
It is composed of an AXI-REG-Interface, a demux, and a set of registers. 
The first stage of the \textit{AXI-REG-Interface} is composed of two FIFOs (of depth $D_{\text{FIFO}}^{\text{IO}}$), buffering read and write transactions, respectively.
After the FIFOs, a round-robin arbiter manages read and write transactions, allowing only one at a time to pass to the protocol conversion.
Since the IO subsystem is meant for low-power reads and writes, registers' transactions share the same set of signals for reads and writes and are limited to single-word accesses.
For such a reason, the IO subsystem does not support burst transactions (requests having $\beta_i > 1$ are suppressed). 
%
%Such a feature simplifies the control logic.
%
\textit{The demux} stage decodes the request and directs it to the proper register destination, where it is finally served as a register read or write.

\subsubsection{Delays and parallelism analysis} The IO subsystem is a \emph{peripheral}, thus, we apply the model proposed in Section~\ref{model:bridge}.
Considering the maximum service delays, overall, the IO subsystem is composed of four stages: (i) the FIFOs, (ii) the protocol conversion, (iii) demultiplexing, and (iv) target register access.
The first three stages, contributing to the control overhead, introduce a fixed delay of one clock cycle ($t_{\text{CK}}^{\text{IO}}$) each for a total of $3 \cdot t_{\text{CK}}^{\text{IO}}$ clock cycles. 
Consider now stage (iv).
In the case of a write, the request and the corresponding data are propagated in parallel in one clock cycle. 
In the case of a read, the register provides the data in the clock cycle following the request -- $t_{\text{CTRL}}^{\text{IO}}$ requires one extra clock cycle.
Summing all the contributions, the service time of the I/O subsystem is upper bounded by:
\begin{equation}
\footnotesize
    t_{\text{CTRL,W}}^{IO} = 3 \cdot t_{\text{CK}}^{\text{IO}} ;  \quad
    t_{\text{CTRL,R}}^{IO} = 4 \cdot t_{\text{CK}}^{\text{IO}} ;  \quad 
    t_{\text{DATA,W/R}}^{IO} = t_{\text{CK}}^{\text{IO}} 
\end{equation}

Consider now the parallelism. Similarly to the SPM module, the IO subsystem is capable of buffering up to $D_{\text{FIFO}}^{\text{IO}}$ of each type in its input FIFO queues. Thus, the maximum number of outstanding transactions supported by the IO subsystem is equal to:
\begin{equation}
\footnotesize
       \chi_{W}^{\text{IO}} = \chi_{R}^{\text{IO}} = D_{\text{FIFO}}^{\text{IO}}
\end{equation}
The IO subsystem serves read and write transactions one at a time, and no pipelining is implemented among the different stages. This means that $\rho^{\text{IO}} = 0$ and $\theta^{\text{SPM}}=0$.

\subsection{The main memory subsystem}\label{analysis:ms}

The main memory subsystem is a \emph{peripheral}: we apply here the model presented in Section~\ref{model:peripheral}. 
It is composed of three macro submodules: (i) the \emph{AXI Last-level Cache (LLC)}; (ii) the \emph{HyperRAM memory controller (HMC)}; and (iii) the \emph{HyperRAM memory (HRAM)}.
It is based on HyperRAM memories leveraging the HyperBUS protocol~\cite{hyperram_low_pincount}.
HyperRAMs are optimized for low-overhead data storage while offering up to 3.2Gbps bandwidth.
HyperRAMs expose a low pin count, a fully digital 8-bit double-data-rate (DDR) interface used for commands and data.
HyperRAMs serve transactions in order, one at a time, as required by the protocol~\cite{hyperram_low_pincount}.
While a pure in-order strategy is simpler than those deployed by high-end commercial memory controllers, it is important to note that these controllers are typically complex closed-source IPs, making detailed analysis extremely challenging. \review{Notably, our analysis is the first to explore this level of detail.}
Furthermore, the memory subsystem under analysis has shown to be effective in tape-outs of Linux-capable chips~\cite{shaheen}.
\review{We model the service times of a single transaction in case of an LLC hit and miss. By doing so, we provide upper bounds that can be leveraged by future studies focusing on LLC interference between different \emph{controllers} at the application level. For example, advanced cache management studies for real-time applications (e.g., cache coloring) could leverage the upper bounds provided here to bound overall task execution times.}

\subsubsection{RTL IP structure}

\begin{figure*}
    \centering
    \includegraphics[width=\textwidth]{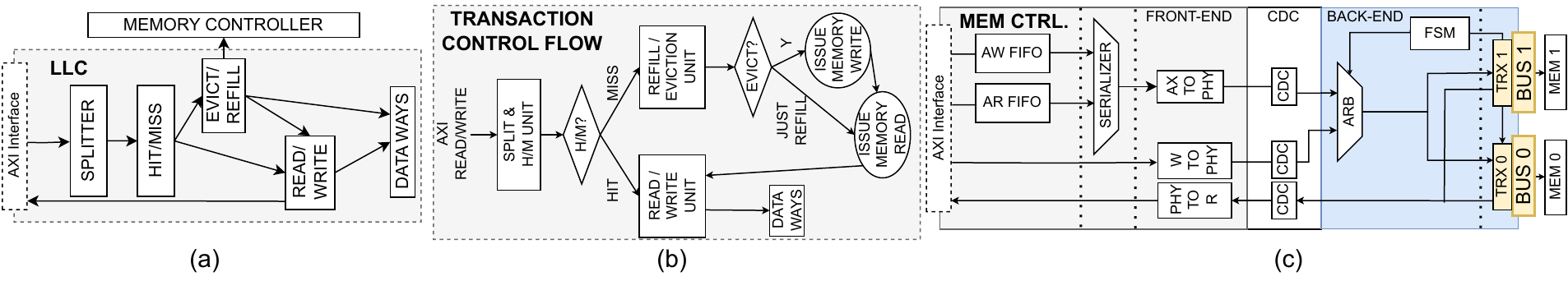}
    \caption{Block diagrams of the components of the main memory subsystem. (a) LLC block diagram, (b) Transaction control flow diagram, (c) Memory controller block diagram.}
    \label{fig:mem-sub-complete}
\end{figure*}

\textbf{\textit{The AXI Last-Level Cache}} is the interface of the memory subsystem with the platform. 
%
%The LLC under analysis has a configurable size, number of ways, and cache line length.
%
The LLC under analysis has configurable cache line length, defined as $LW_{\text{LLC}}$.
Figure~\ref{fig:mem-sub-complete}(a) shows the LLC's block diagram, composed of 5 pipelined units: (i) burst splitter, (ii) hit-miss detection, (iii) eviction/refill, (iv) data read/write, and (v) data ways. 
Figure~\ref{fig:mem-sub-complete}(b) shows how these units cooperate to serve the requests.
The burst splitter buffers and splits the incoming AXI requests into multiple sub-requests that have the same length of the cache line, and it calculates the tags of the sub-transactions.
A $\beta_i$-word AXI burst request is split internally into $\lceil \frac{\beta_i}{LW_{\text{LLC}}} \rceil$ requests of length $LW_{\text{LLC}}$.
The tags are the input to the hit-miss detection unit, which analyzes them to determine if any sub-request will be a (a) hit or (b) miss. 
In case (a), the transaction is directed to the read/write unit: if it is a (a.i) read, the read response is generated and immediately sent through the AXI subordinate port, completing the transaction.
In the case of a (a.ii) write, the locally cached value is updated, and a write response is generated and sent back to the AXI interface to complete the transaction.
In case (b), the transaction is submitted to the eviction/refill unit.
Refill is performed on every miss and consists of issuing a read to the memory controller to fetch the missing data and update the data way. 
Eviction is performed when a cache set is full to free the necessary spot before a refill. A Least Recently Used (LRU) algorithm is used in the module under analysis. 

\textbf{\textit{The HyperRAM memory controller}}~\cite{hyper-git} is depicted in Figure~\ref{fig:mem-sub-complete}(c).
It consists of two tightly coupled modules working in two separated frequency domains: (i) the AXI \emph{front-end} and (ii) the \emph{back-end} PHY controller.  
The front-end handles and converts the AXI transactions into data packets for the PHY controller; it runs at the same clock as the LLC ($t_{\text{CK}}^{\text{HMC}}$).
The back-end features a Finite State Machine (FSM) to send/receive the data packets and keep compliance with the HyperBUS protocol timings and data flow; it runs at the same clock as the HyperRAMs ($t_{\text{CK}}^{\text{HRAM}}$).
The back-end handles two off-chip HyperRAMs in parallel, configured with interleaved addresses. 
As each HyperRAM arranges data as 16-bit words, the word size of the back-end is $DW_{\text{HYPER}}=32$ bits.

The first stage of the front-end is composed of two FIFOs buffering incoming AXI read and write requests.
Then, a serializer solves conflicts among reads and writes, allowing only one AW or AR request at a time. 
Following, three modules translate between AXI and the back-end protocol: (i) AXTOPHY, translating the AXI AW or AR requests into commands for the back-end; (ii) PHYTOR converting the data words from the back-end into AXI read beats for the AXI interface; and (iii) WTOPHY, converting AXI W data beats into data words and generating write response at the end of the transaction. 
Three CDC FIFOs are deployed between the AXTOPHY, WTOPHY, and PHYTOR and the back-end.
The back-end deploys an internal FSM arranging the requests coming from the front-end into 48-bit vector requests, as required in the HyperBUS protocol, and propagating the data packets to/from the two physical HyperRAM memories through two \textit{transceivers} (TRX).

\textbf{\textit{The HyperRAM memory}} is an off-chip memory IP~\cite{hyperram_low_pincount}. It is provided with a cycle-accurate model, fundamental for our analysis purposes~\cite{hyper-model}.
Each HyperRAM is organized as an array of 16-bit words and supports one outstanding burst transaction, up to 1kB long.
As two HyperRAM are interleaved, the overall burst can be up to 2kB long\cite{hulk-v}.

\subsubsection{Delays and parallelism analysis}

We now bound the worst-case service time of the main memory subsystem, analyzing its components one at a time.
Starting with the LLC, we follow the control flow diagram reported in Figure~\ref{fig:mem-sub-complete}(b) to guide the explanation.
The LLC collects the requests incoming to the main memory.
Three scenarios can happen: (i) LLC cache hit, (ii) LLC cache miss with refill, and (iii) LLC cache miss with eviction and refill.

In case (i), the LLC directly manages the request, and no commands are submitted to the HMC.
The request proceeds through the LLC splitter, hit/miss unit, read/write unit, and data way stages. By design, each stage of the LLC requires a fixed number of clock cycles.
The burst splitter executes in one clock cycle ($t_{\text{CK}}^{\text{LLC}}$).
The hit/miss detection stage takes two clock cycles ($2t_{\text{CK}}^{\text{LLC}}$): one for tag checking and one to propagate the request to the read/write unit or the evict/refill unit.
The read/write unit requires one clock cycle ($t_{\text{CK}}^{\text{LLC}}$) to route the transaction to the data ways.
The data ways accept the incoming request in one clock cycle ($t_{\text{CK}}^{\text{LLC}}$) to then access the internal SRAM macros (same as the SPM, Section~\ref{ssec:l2spm}).
The internal SRAM takes one clock cycle to provide the read data ($t_{\text{CK}}^{\text{LLC}}$), but no further latency is required on writes.
Once it gets the response, the read/write unit routes the read channel to the AX interface, whereas it takes one clock cycle ($t_{\text{CK}}^{\text{LLC}}$) to generate the write B response at the end. 
Thus, read/write unit and data ways take together three clock cycles ($3t_{\text{CK}}^{\text{LLC}}$).
Summing up the contributions, the service time in case of a hit is upper bound by: 
\begin{equation}
\footnotesize
    t_{\text{CTRL,R/W}}^{\text{MS-HIT}} = 6 \cdot t_{\text{CK}}^{\text{LLC}};  \quad 
    t_{\text{DATA,R/W}}^{\text{MS-HIT}} = t_{\text{CK}}^{\text{LLC}};
\end{equation}
%
%Here, the LLC accounts both for control and data delays with the data ways SRAMs contributing to the data service time.

Consider now cases (ii) and (iii): the eviction and refill stage is also involved, and a read (for refill) and, optionally, a write (for eviction) is issued to the main memory.
Eviction and refill are run in parallel.
Each operation performs two steps, each taking one clock cycle: (a) generating a transaction for the main memory and (b) generating a transaction for the data way.
Thus, summing the latency introduced by the eviction and refill stage ($2t_{\text{CK}}^{\text{LLC}}$) with the ones from the other stages, the LLC's contribution to the overall control time in case of a miss is upper bound by:
\begin{equation}
\footnotesize
t_{\text{CTRL,R/W}}^{\text{LLC-MISS}} = t_{\text{CTRL,R/W}}^{\text{MS-HIT}} + 2 t_{\text{CK}}^{\text{LLC}}
\end{equation}

Consider now the delay introduced by the HMC on a generic request. 
Later, we will use it to bound the service time for the batch of transactions issued by the LLC.
As described earlier, the HMC is composed of (a) the front-end, (b) the CDC FIFOs, and (c) the back-end.  
Consider (a): each one of the front-end's submodules takes one clock cycle to sample and process the transaction, except for the serializer, which takes two. 
As transactions pass through 4 modules (FIFOs, serializer, AXITOPHY, and either WTOPHY or PHYTOR), the overall delay contribution of the front-end is equal to $5t_{\text{CK}}^{\text{HMC}}$. 
Consider now (b): these are the CDC FIFOs composing the AXI CDC FIFOs introduced in Section~\ref{ssec:axi_cdc_fifo}. 
For writes, the transmitter (TX) is the front-end, sending data to the back-end from the AXTOPHY and the WTOPHY.
As both transfers happen in parallel, the delay introduced by the CDC on a write is upper bound by $t_{\text{CDC}}(t_{\text{CK}}^{\text{HMC}},t_{\text{CK}}^{\text{HRAM}})$.
For reads, first, the front-end transmits (TX) the AXTOPHY request, and then the back-end transmits the data beats: the delay introduced by the CDC on a read is upper bound by $t_{\text{CDC}}(t_{\text{CK}}^{\text{HMC}},t_{\text{CK}}^{\text{HRAM}}) + t_{\text{CDC}}(t_{\text{CK}}^{\text{HRAM}},t_{\text{CK}}^{\text{HMC}})$. 
Consider now (c): the back-end's FSM parses the incoming request into a HyperRAM command in one cycle ($t_{\text{CK}}^{\text{HRAM}}$).
Following this, an extra cycle is required for the data to cross the back-end.
Summing up the contributions just described, the control time of the HMC on a generic transaction is upper bound by:
\begin{equation}
\footnotesize
\begin{split}
& t^{\text{HMC}}_{\text{CTRL,R}} = 5 \cdot t_{\text{CK}}^{\text{HMC}} + t_{\text{CDC}}(t_{\text{CK}}^{\text{HMC}},t_{\text{CK}}^{\text{HRAM}}) + t_{\text{CDC}}(t_{\text{CK}}^{\text{HRAM}},t_{\text{CK}}^{\text{HMC}}) + 2 \cdot t_{\text{CK}}^{\text{HRAM}} \\
& t^{\text{HMC}}_{\text{CTRL,W}} = 5 \cdot t_{\text{CK}}^{\text{HMC}} + t_{\text{CDC}}(t_{\text{CK}}^{\text{HMC}},t_{\text{CK}}^{\text{HRAM}}) + 2 \cdot t_{\text{CK}}^{\text{HRAM}}
\end{split}
\end{equation}

Consider now the delays introduced by the HyperRAM memories on a generic request.
The control overhead time to access the HyperRAM memory is defined by the HyperBUS protocol~\cite{hyperram_low_pincount}.
First, the 48-bit HyperRAM command vector is sent over the two memories in $3\cdot t_{\text{CK}}^{\text{HRAM}}$ clock cycles, as the HyperBUS command bus is 16 bits.
Following, the HyperBUS provides a fixed latency for the maximum time to access the first data word, accounting for refresh effects and crossing row boundaries.
The specifications~\cite{hyper-spec} bound such a delay between 7 and 16 clock cycles. 
In our case, this is set to $12 \cdot t_{\text{CK}}^{\text{HRAM}}$. Thus, the total control latency of the HyperRAM memory is upper bound by: 
\begin{equation}
\footnotesize
    t_{\text{CTRL,R/W}}^{\text{HRAM}}=15 \cdot t_{\text{CK}}^{\text{HRAM}} 
\end{equation}
At this point, data are ready to be propagated. 
As the AXI domain and the HyperRAM have different data widths, the number of cycles to send/receive an AXI word is: 
\begin{equation}
    \footnotesize
        t_{\text{DATA,R/W}}^{\text{HRAM}} = DW_{\text{HYPER}} \cdot \lceil \frac{DW_{\text{AXI}}}{DW_{\text{HYPER}}} \rceil \cdot t_{\text{CK}}^{\text{HRAM}}
\end{equation}

We now have all the elements to bound the overall service time of the whole main memory subsystem in case of a miss (ii) with refill and (iii) eviction and refill.
First, we bound the service time to serve a refill (read) request.
A $\beta_i$-long transaction is split by the LLC into $\lceil \beta_i/LW_{\text{LLC}} \rceil$ sub-transactions to the memory, each $LW_{\text{LLC}}$-long.
Therefore, by multiplying the control time of each sub-transaction ($t_{\text{CTRL,R}}^{\text{HMC}} + t_{\text{CTRL,R}}^{\text{HRAM}})$ by the number of transactions issued ($\lceil\frac{\beta_i}{LW_{\text{LLC}}}\rceil$), we bound the control time introduced by the memory controller and the off-chip memories.
To this, we sum the control time of the LLC in case of a miss ($t_{\text{CTRL,W/R}}^{\text{MS-MISS}}$) and obtain the whole control overhead.
The same reasoning applies to the data time: the total number of values requested by the LLC to the memory will be equal to $LW_{\text{LLC}} \cdot \lceil \frac{\beta_i}{LW_{\text{LLC}}}\rceil $ and the overall time spent reading $LW_{\text{LLC}} \cdot \lceil \frac{\beta_i}{LW_{\text{LLC}}}\rceil t_{\text{DATA,R/W}}^{\text{HRAM}}$.
It follows that the time to serve one word is $\frac{LW_{\text{LLC}}}{\beta_i} \cdot \lceil \frac{\beta_i}{LW_{\text{LLC}}}\rceil \cdot t_{\text{DATA,R/W}}^{\text{HRAM}}$.
Summing it with the data time of the LLC ($t_{\text{DATA,R/W}}^{\text{MS-HIT}}$), we obtain the following upper bounds for case (ii):
\begin{equation}\label{eq:miss-refill}
\footnotesize
\begin{split}
    & t_{\text{CTRL,R/W}}^{\text{MS-MISS-REF}} = t_{\text{CTRL,R}}^{\text{LLC-MISS}} +  \left\lceil\frac{\beta_i}{LW_{\text{LLC}}}\right\rceil \cdot ( t_{\text{CTRL,R}}^{\text{HMC}} + t_{\text{CTRL,R}}^{\text{HRAM}} ) ; \\
    & t_{\text{DATA,R/W}}^{\text{MS-MISS-REF}} = t_{\text{DATA,R/W}}^{\text{MS-HIT}} +  \frac{LW_{\text{LLC}}}{\beta_i} \cdot \left\lceil\frac{\beta_i}{LW_{\text{LLC}}}\right\rceil \cdot t_{\text{DATA,R}}^{\text{HRAM}} ; 
    \end{split}
\end{equation}
If the eviction is also required, $\lceil \frac{\beta_i}{LW_{\text{LLC}}} \rceil$ extra write transactions of length $\beta_i$ are performed to save the evicted data.
Following the same reasoning as earlier, this batch of transactions will introduce $\lceil\frac{\beta_i}{LW_{\text{LLC}}}\rceil  (t_{\text{CTRL,W}}^{\text{HMC}} + t_{\text{CTRL,W}}^{\text{HRAM}})$ clock cycles to the control time and $\frac{LW_{\text{LLC}}}{\beta_i} \cdot\lceil\frac{\beta_i}{LW_{\text{LLC}}}\rceil \cdot t_{\text{DATA,W}}^{\text{HRAM}}$ to the data time.
We sum these numbers to eq. \ref{eq:miss-refill} to upper bound the overall control and data time as follows:
\begin{equation}
\footnotesize
\begin{split}
    t_{\text{CTRL,W/R}}^{\text{MS-MISS-REF-EV}} = t_{\text{CTRL,W/R}}^{\text{MS-MISS-REF}} +  \left\lceil\frac{\beta_i}{LW_{\text{LLC}}}\right\rceil  ( t_{\text{CTRL,W}}^{\text{HMC}} + t_{\text{CTRL,W}}^{\text{HRAM}} ) ; \\
    t_{\text{DATA,W/R}}^{\text{MS-MISS-REF-EV}} = t_{\text{DATA,W/R}}^{\text{MS-MISS-REF}} +  \frac{LW_{\text{LLC}}}{\beta_i} \cdot \left\lceil\frac{\beta_i}{LW_{\text{LLC}}}\right\rceil \cdot t_{\text{DATA,W}}^{\text{HRAM}} ; \\
\end{split}
\end{equation}

Consider now the parallelism of the main memory subsystem. This is defined by the LLC, which acts as an interface with the rest of the platform, buffering up to $D_{\text{FIFO}}^{\text{LLC}}$ read and write transactions. This means that the maximum number of supported outstanding transactions is as follows:
\begin{equation}
\footnotesize
\chi_R^{MS} = \chi_W^{MS} = D_{\text{FIFO}}^{\text{LLC}}
\end{equation}
The LLC is pipelined: in the case all the enqueued accesses are hits, there is no stalling.
However, the memory controller handles only one transaction at a time, stalling the preceding ones, and only serves one read or one write at a time.
Hence, as soon as one access is a miss, $\rho^{\text{MS}} = 0$ and $\theta^{\text{MS}}=0$.

\subsection{AXI host crossbar}\label{analysis:xbar}

The AXI host crossbar under analysis is a consolidated AXI crossbar already validated in multiple silicon tapeouts~\cite{shaheen},\cite{cheshire}, \cite{Kurth_2021}.
We apply here the generic model for the \emph{crossbar} proposed in Section~\ref{model:router}. 
The crossbar is referred as $R_0$.

\subsubsection{RTL IP structure}

As detailed in Figure~\ref{fig:xbar_uarch}, the crossbar exports a set of input subordinate ports (S) and output manager ports (M).
Each S port is connected to a demultiplexer, which routes the incoming AW and AR requests and W data to the proper destination. 
Each M port is connected to a multiplexer, which (i) arbitrates AW and AR requests directed to the same \emph{peripheral}, (ii) connects the selected W channel from the \emph{controller} to the \emph{peripheral}, and (iii) routes back the R read data and B write responses.
The crossbar under analysis can be configured for a fully combinatorial (i.e., decoding and routing operations in one clock cycle) or pipelined structure with up to three pipeline stages.
In the platform under analysis, it is configured to be fully combinatorial.

\subsubsection{Delays and parallelism analysis} 
To analyze the maximum propagation delays introduced by the crossbar, we upper bound the overall latency on a transaction by combining the delays introduced on each AXI channel.
%
%It is worth mentioning that, differently from the other analyzed components, the maximum delays introduced at the crossbar level depend on the interference generated by the other \emph{controllers}.
%
We provide two upper bounds, one for transactions in isolation (i.e., $t_{\text{PROP,R/W}}^{R_0}$ as defined in Section \ref{sec:model}) and the other for transactions under contention (i.e., $t_{\text{PROP,R/W}}^{R_0} + t_{\text{CON,R/W}}^{R_0}$ as defined in Section \ref{sec:model}). We will use both of them in our architectural analysis reported in Section~\ref{sec:system-level-analysis}.

\textit{Maximum delays in isolation:} Thanks to the combinatorial structure, it is guaranteed by design that a request for a transaction, a data word, or a write response crosses the crossbar in one clock cycle ($t_\text{CK}^{{R_0}}$). 
Consider a whole AXI transaction.
For a read transaction, the crossbar is crossed twice: on the AR and R AXI channels, respectively.
For each AXI write transaction, the crossbar is crossed two times: the first time is crossed by the AW and W beats (propagated in parallel), and the second time by the B response. 
Thus, the propagation delays in isolation are equal to:
\begin{equation}\label{eq:xbar_prop}
\footnotesize
    t_{\text{PROP,R/W}}^{{R_0}} = 2 \cdot t_{\text{CK}}^{{R_0}};
\end{equation}
\begin{figure}[t]
    \centering
    \includegraphics[width=0.9\linewidth]{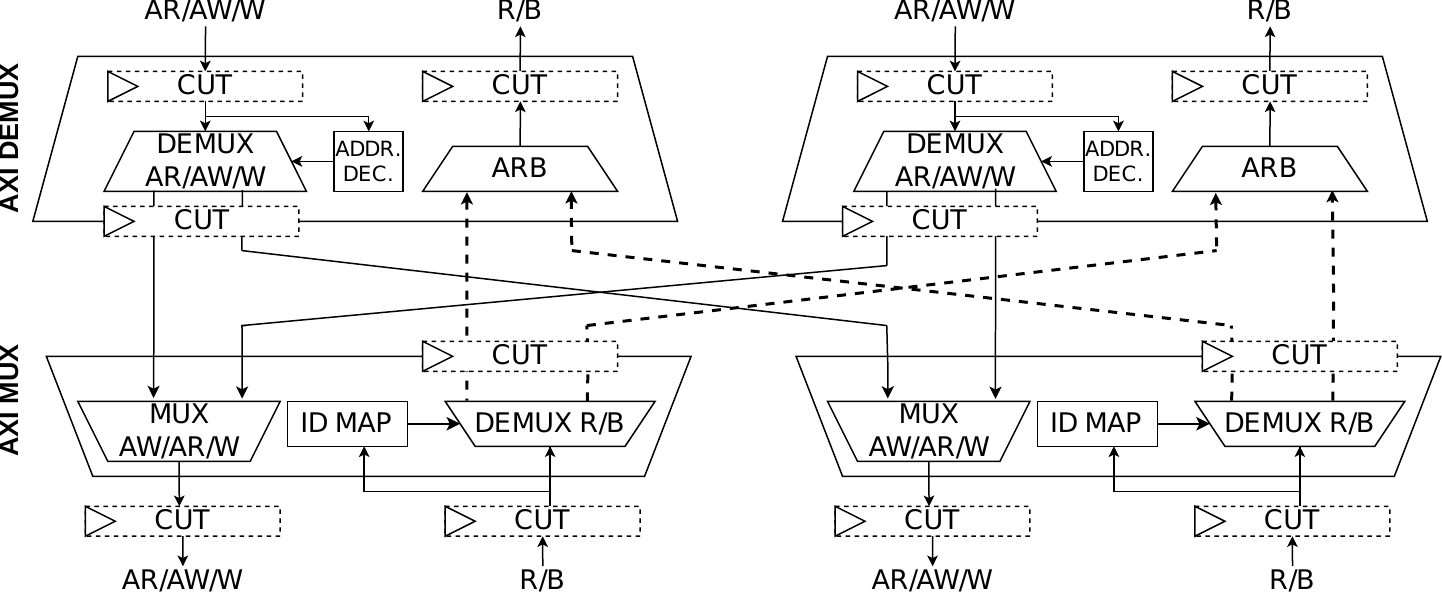}
    \caption{AXI Crossbar block diagram}
    \label{fig:xbar_uarch}
\end{figure}
\textit{Maximum delays under contention:}
Under contention, multiple \emph{controllers} connected to the crossbar can attempt to concurrently send requests to the same \emph{peripheral}, generating interference.
The arbiters deploy a round-robin scheme capable of granting one AW and one AR request for each clock cycle.
In the worst-case scenario, the request under analysis loses the round-robin and is served last, experiencing a delay of $M_{{R_0}}-1$ clock cycles (with $M_{{R_0}}$ the number of \emph{controller} capable of interfering with the request under analysis).
%
%No interference is experienced at the W, R, and B channels.
%
%The strict ordering imposed by AXI mandates that the crossbar locks the W data and B response channels for the request under service (no arbitration is performed). 
%
%Consider now read transactions, contention can arise when a controller targets more peripherals simultaneously, i.e., multiple peripherals can concurrently send back to a controller under analysis.
%
%However, this contention is auto-generated: the demultiplexer merges the responses of the multiple R channels into the controller R channel without adding any latency itself.
%
From the previous considerations, the maximum propagation time introduced by the crossbar is upper bound by:
\begin{equation}\label{eq:xbar_contention}
\footnotesize
 t_{\text{CON,R}}^{{R_0}}  = t_{\text{CON,W}}^{{R_0}} = M_{{R_0}}-1
\end{equation}

Consider now the parallelism. 
Concerning reads, the crossbar does not keep track of the inflight transactions. To route the responses back, it appends information to the AXI ID. Doing so does not limit the maximum number of outstanding transactions.
The behavior is different for writes: AXI enforces a strict in-order execution of write transactions (see ~\cite{ARMAXI} p. 98).
This requires the crossbar to implement a table to know the order of granted transactions.
The maximum number of outstanding write transactions per S port is limited by the depth of such tables, refereed as $D_{\text{TAB}}^{{R_0}}$. 
From the previous consideration: $\chi_{W}^{{R_0}}=D_{\text{TAB}}^{{R_0}}$. 
In the architecture under analysis, $\chi_{W}^{{R_0}}$ is set to be bigger than the parallelism supported by the \emph{peripherals} so that the crossbar does not limit the overall parallelism of the system.

\section{System-level worst-case response time analysis}\label{sec:system-level-analysis}

This section introduces our system-level end-to-end analysis to upper bound the overall response times of read and write transactions issued by a generic \emph{controller} and directed to a generic \emph{peripheral}, considering the maximum interference generated by the other \emph{controllers} in the system. Our approach is static\cite{thewcetproblem} and compositional\cite{mitraDandT}. Specifically, we leverage the component-level static analysis introduced in Section \ref{sec:wcea} to then compose, step-by-step, the system-level worst-case service time of transactions traversing the whole architecture.
% 

% \subsection{Independence between \emph{peripherals}}
%
We make an assumption aligned with the SoA~\cite{restuccia2020axi,restuccia2020modeling,hassanzoni,jiang2022bluescale,hassan,mirosanlou2020mcsim} to ensure independence among \emph{peripherals} while not compromising the generality of the analysis. It is assumed that multiple outstanding transactions of the same type (either read or write) issued by the same \emph{controller} target the same \emph{peripheral}: before issuing a transaction targeting a \emph{peripheral} $P_j$, a \emph{controller} completes the pending transactions of the same type targeting a different \emph{peripheral} $P_z$.
Without such an assumption, due to the strict ordering imposed by the AXI standard~\cite{ARMAXI} on the W channel, and the structure of some \emph{peripherals} generating interference between reads and writes (i.e., $\rho^{P_j}=0$), transactions issued by $C_k$ and directed to $P_j$ might interfere with transactions issued by $C_i$ and directed to $P_z$, if $C_i$ also issues in parallel transactions to $P_j$, and vice-versa. 
This assumption allows us to relax our analysis, removing such pathological cases. It is worth noticing that it does not enforce any relationship between read and write transactions.
Such an assumption can either be enforced at the software level or at the hardware level.
The results of our analysis can be extended to such corner cases if required.
We leave this exploration for future works. 

The first step of the analysis is to bound the overall response time of a transaction in isolation (Lemma~\ref{lemma:st_single_iso}).
Secondly, we bound the maximum number of transactions that can interfere with a transaction under analysis, either of the same type (e.g., reads interfering with a read, Lemma~\ref{lemma:max_num_intf_read}) or of a different type (e.g., write interfering with a read, and vice versa, Lemma~\ref{lemma:intf_diff_kind}).
Lemma~\ref{lemma:max_intf_1trx} bounds the maximum temporal delay each interfering transaction can delay a transaction under analysis. 
Finally, Theorem~\ref{theorem:max_intf} combines the results of all the lemmas to upper bound the overall worst-case response time of a transaction under analysis under interference.
We report the lemmas in a general form.
$AX_{i,j}$ can represent either a read or write transaction issued by the generic \emph{controller} $C_i$ and directed to the generic \emph{peripheral} $P_j$. The \emph{crossbar} is referred to as $R_0$. To make our analysis general, we assume that $\Psi_j = [C_0, ..., C_{M-1}]$ is the generic set of interfering \emph{controllers} capable of interfering with $C_i$ issuing transactions to $P_j$ and that that a generic set of \emph{bridges} $\Theta_{i} = \{Q_0, ..., Q_{q-1} \}$ can be present between each \emph{controller} $C_i$ and the crossbar $R_0$. 
The cardinality of $\Psi_j$ is referred to as $\bigm|\Psi_j\bigm|$ and corresponds to the number of \emph{controllers} interfering with $AX_{i,j}$.

\begin{lemma}\label{lemma:st_single_iso}
    %Consider a generic transaction $AX_{i,j}$. To reach $P_j$, the transaction traverse a generic set of \emph{bridges} $\Theta_{i,j} = \{Q_0, ..., Q_{q-1} \}$, placed between $C_i$ and the Router $R_0$. 
    The response time in isolation of $AX_{i,j}$ is upper bounded by:
    \begin{equation}\label{eq:one_trx_iso}
    \footnotesize
    d^{X}_{i,j} = d_{\text{R/W}}^{P_j}  + \sum_{Q_l \in \Theta_{i,j}} d_{\text{R/W}}^{Q_l} + d_{\text{R/W}}^{R_0} 
    \end{equation}
\end{lemma}
\begin{proof}
    Section~\ref{sec:wcea} upper bounds the worst-case delays in isolation introduced by each component in the platform. According to their definition, such delays account for all of the phases of the transaction. The components in the platform are independent of each other. Thus, the delay introduced by each traversed component is independent of the behavior of the other components.
    It derives that the overall delay incurred in traversing the set of components between $C_i$ and $P_j$ is upper bounded by the sum of the worst-case delays introduced by all of the components in the set. 
    Summing up the maximum delay introduced by the target \emph{peripheral} $P_j$ ($d_{\text{R/W}}^{P_j}$), by the set of traversed \emph{bridges} $\Theta_{i}$, and by the \emph{crossbar} $R_0$ ($d_{\text{R/W}}^{R_0}$), the lemma derives.
\end{proof}

\begin{lemma}\label{lemma:max_num_intf_read}
    The maximum number of transactions of the same type that can interfere with $AX_{i,j}$ is upper bounded by:    
\begin{equation}\label{read_with_read}
\footnotesize
S^{X}_{i,j} = min\left(\sum_{C_y \in \Psi_j} \phi_X^{C_y}, \chi_X^{P_j} + \bigm|\Psi_j\bigm| \right)
\end{equation}
\end{lemma}

\begin{proof}
    The $\min$ in the formula has two components.
    As from the AXI standard definition, an interfering \emph{controller} $C_k$ cannot have more than $\phi_{X}^{C_k}$ pending outstanding transactions. This means that summing up the maximum number of outstanding transactions for each interfering \emph{controller} in $\Psi_j$ provides an upper bound on the number of transactions of the same type interfering with $AX_{i,j}$ -- the left member of the $\min$ derives.
    From our \emph{peripheral} analysis reported in Section~\ref{sec:wcea}, $P_j$ and $R_0$ can limit the maximum amount of transactions accepted by the system: $P_j$ accepts overall at most $\chi_{R/W}^{P_j}$ transactions -- when such a limit is reached, any further incoming transaction directed to $P_j$ is stalled.
    After $P_j$ serves a transaction, $R_0$ restarts forwarding transactions to the \emph{peripheral} following a round-robin scheme (see Section~\ref{sec:wcea}).
    In the worst-case scenario, $C_i$ loses the round-robin arbitration against all of the $\bigm|\Psi_j\bigm|$ interfering \emph{controllers} in $\Psi_j$, each ready to submit an interfering request. 
    Summing up the contributions, also $\chi_{R}^{P_j} + \bigm|\Psi_j\bigm|$ upper bounds the maximum number of transactions interfering with $AX_{i,j}$ -- the right member of the $\min$ derives.
    Both of the bounds are valid -- the minimum between them is an upper bound providing the least pessimism -- Lemma \ref{lemma:max_num_intf_read} derives.
\end{proof}

\begin{lemma}\label{lemma:intf_diff_kind}
The maximum number of transactions of a different type (represented here as Y, i.e., write transactions interfering with a read under analysis, and vice versa) interfering with $AX_{i,j}$ is upper bounded by: 
%
%\begin{equation}
%\begin{cases}\label{write_with_read}
%\footnotesize
%        U^Y_{i,j} =  S^X_{i,j} + 1 \quad & when \quad \theta^{P_j}=0 \\
%        U^Y_{i,j} = 0  \quad & when \quad  \theta^{P_j}=1 
%\end{cases}  
%\end{equation}
\begin{equation}\label{write_with_read}
\footnotesize
        U^Y_{i,j} =  ( S^X_{i,j} + 1 ) \cdot ( 1 - \theta^{P_j} ) \\
\end{equation}
\end{lemma}
\begin{proof}
    % We provide proof on a read transaction, i.e., equation \ref{write_with_read}. As read and write requests undergo the same arbitration scheme, the case for write (eq. \ref{read_with_write}) follows. 
    %
    According to Section~\ref{analysis:xbar}, $R_0$ manages transactions of different types independently -- thus, no interference of this type is generated at the $R_0$ level.
    From Section \ref{sec:model}, $\theta^{P_j}=1$ represents the case in which the \emph{peripheral} is capable of serving read and write transactions in parallel (e.g., the SPM \emph{peripheral}, Section \ref{ssec:l2spm}). Thus, no interference is generated among them -- the second equation derives.
    From Section \ref{sec:model}, $\theta^{P_j}=0$ represents the case in which $P_j$ does not feature parallelism in serving read and write transactions (i.e., also write transactions interfere with reads, e.g., main memory subsystem, Section ~\ref{analysis:ms}). 
    Considering lemma~\ref{lemma:max_num_intf_read}, at most $S^X_{i,j}$ transactions of the same type can interfere with $AX_{i,j}$. 
    With $\theta^{P_j}=0$, and assuming a round-robin scheme arbitrating between reads and writes at the \emph{peripheral} level, each one of the $S^X_{i,j}$ interfering transaction of the same type can be preceded by a transaction of the opposite type, which can, therefore, create interference.
    The same applies to $AX_{i,j}$, which can lose the arbitration at the \emph{peripheral} level as well. 
    Summing up the contribution, it follows that $S^X_{i,j} + 1$ can overall interfere with $AX_{i,j}$ -- the first equation derives. 
\end{proof}

\begin{lemma}\label{lemma:max_intf_1trx}
    The maximum time delay that a transaction of any kind $AX_{k,j}$ issued by the generic interfering \emph{controller} $C_k$ can cause on $AX_{i,j}$ is upper bounded by:
    % Consider a transaction $AX_{i,j}$ (either read or write) issued by the \emph{controller} $C_i$ and directed to the \emph{peripheral} $P_j$ connected to $R_0$. Consider $C_k$ as an interfering \emph{controller} issuing a transaction $AX_{k,j}$ of burst length $\beta_k$ directed to $P_j$, and traversing the set of \emph{bridges} $\Theta_{k,x}=[Q_0,...,Q_{q-1}]$, between $C_k$ and $R_0$.
    %
    %
\begin{equation}\label{eq:max_intf_1trx}
\footnotesize
     \Delta_{k,j} = d_{\text{R/W}}^{R_0} + (1 - \rho^{P_j}) \cdot t_{\text{CTRL,R/W}}^{P_j} + t_{\text{DATA,R/W}}^{P_j}\cdot \beta_k 
\end{equation}
\end{lemma}

\begin{proof}
    In traversing the path between $C_k$ and $P_j$, $AX_{k,j}$ shares a portion of the path with $AX_{i,j}$, i.e., the target \emph{peripheral} $P_j$ and the crossbar $R_0$ -- no \emph{bridges} from $\Theta_{k}$ belongs to the shared path, thus the delay propagation of $AX_{k,j}$ do not contribute in delaying $AX_{k,j}$.
    Considering the delay generated by $AX_{k,j}$ at $R_0$, this is upper bounded by $d_{\text{R/W}}^{R_0}$ in Section~\ref{model:router}.
    As from Section~\ref{model:peripheral}, $t_{\text{CTRL,R/W}}^{P_j} + t_{\text{DATA,R/W}}^{P_j} \cdot \beta_k$ is the maximum service time of $P_j$ for the transaction $AX_{k,j}$ and upper bounds the maximum temporal delay that $AX_{k,j}$ can cause on $AX_{i,j}$ at $P_j$.
    As from the definition of an interfering transaction, $AX_{k,j}$ is served by $P_j$ before $AX_{i,j}$.
    As defined by the model in Section~\ref{model:peripheral}, when $\rho^{P_j}=1$, the \emph{peripheral} works in a pipeline fashion.
    This means that for $\rho^{P_j}=1$, the control time $t_{\text{CTRL,R/W}}^{P_j}$ of an interfering transaction is pipelined and executed in parallel with the transaction under analysis.
    Differently, when $\rho^{P_j}=0$, no pipeline is implemented, and the control time of the interfering transaction can partially or totally interfere with the transaction under analysis. 
    From the previous considerations, the contribution $(1 - \rho^{P_j}) \cdot t_{\text{CTRL,R/W}}^{P_j} + t_{\text{DATA,R/W}}^{P_j}\cdot \beta_k$ derives.
    Summing up the contributions, the lemma follows.
\end{proof}

\begin{theorem}\label{theorem:max_intf}
        % Consider a transaction $AX^{i,j}$ issued by the \emph{controller} $C_i$ and directed to the \emph{peripheral} $P_j$ connected to $R_0$ and $\Psi_j = [C_0, ..., C_{w-1}]$ the set of \emph{controllers} able to reach \emph{peripheral} $P_j$, excluding $C_i$.
        %
        The overall response time of $AX^{i,j}$ under the interference generated by the other \emph{controllers} in the system is upper bounded by:
        \begin{equation}\label{eq:theorem}
\footnotesize
H^X_{i,j} = d^X_{i,j} + (S^X_{i,j} + U^Y_{i,j}) \cdot \Delta_{k,j}
        \end{equation}    
\end{theorem}

\begin{proof}
Summing up the contribution in isolation for $AX_{i,j}$ (Lemma~\ref {lemma:st_single_iso}) with the sum of the maximum number of interfering transactions of the same type (Lemma~\ref{lemma:max_num_intf_read}) and of a different type (Lemma~\ref{lemma:intf_diff_kind}) multiplied by the maximum delay generated by each interfering transaction (Lemma~\ref{lemma:max_intf_1trx}), Theorem \ref{theorem:max_intf} derives.
\end{proof}

The results presented in this Section represent analytical upper bounds derived through static code analysis and the formulation of mathematical proofs. 
Section \ref{sec:exps} will validate them through a comprehensive set of cycle-accurate experiments and measurements.

\section{Experimental validation}\label{sec:exps}

This Section describes the experimental campaign we conducted to validate the methodology and models.
The aim of the experimental campaign is to assess that the results presented in the previous Sections correctly upper bound the maximum delays and response times at the component level and the architectural level.
We follow a hierarchical approach: at first, Section~\ref{ss:exper-IP-latencies} aims to validate the results at the component level we proposed in Section~\ref{sec:wcea}.
Following, in Section~\ref{ss:exper-system}, we experimentally validate the system-level analysis we proposed in Section~\ref{sec:system-level-analysis}. 
The experiments are conducted in a simulated environment (leveraging the Siemens QuestaSIM simulator) and by deploying the design on an FPGA platform.
In the simulated experiments, we deploy custom AXI managers for \textit{ad-hoc} traffic generation and cycle-accurate performance monitors. The generic custom manager represents a generic configurable \emph{controller} $C_i$ issuing requests for transactions -- we will refer to that as $GC_i$.
In the FPGA, we leverage CVA6 and the PULP cluster to generate the traffic \review{with synthetic software benchmarks} and rely on their performance-monitoring registers to collect the measurements. The experimental designs are deployed on the AMD-Xilinx VCU118, using the Vitis 2022.1 toolchain.
\review{Similar State-of-the-Art works upper bounding the execution time of a single transaction leverage synthetic benchmarks to measure the worst-case access times since generic applications fail to do so\cite{restuccia2022bounding,hassanzoni,wu2023ditty}.}
\review{For this reason, we concentrate on synthetic benchmarks at the IP and the system level.}

\subsection{Component-level hardware modules}\label{ss:exper-IP-latencies}

\begin{figure*}[t]
    \centering
    \includegraphics[width=\textwidth]{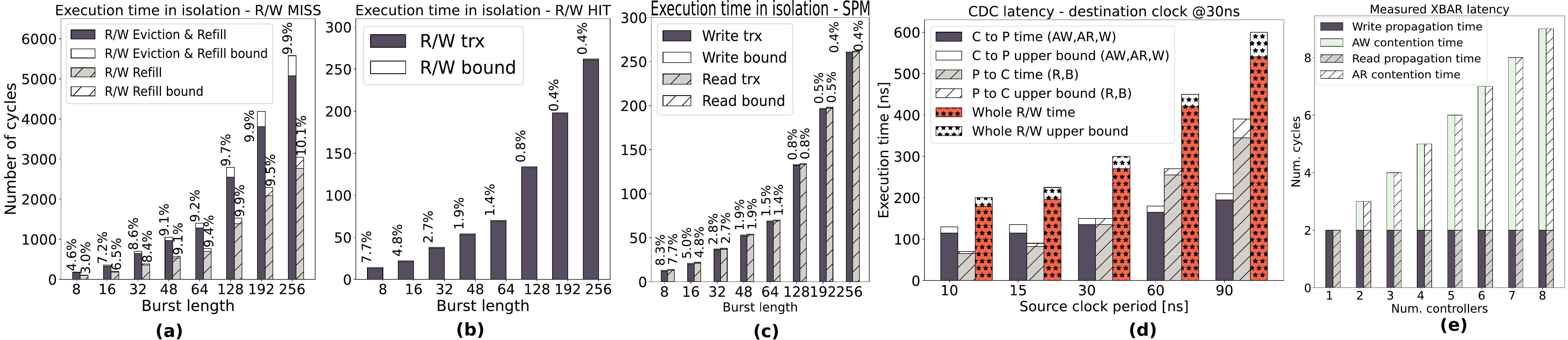}
    \caption{Services time in isolation.}
    \label{fig:experiments}
\end{figure*}

%\subsection{Peripherals experimental analysis}
%\subsubsection{Latencies}\label{ss:exper-IP-latencies}
%
\subsubsection{Delays analysis}
\review{This subsection presents the tests run to measure the worst-case access latency time in isolation for the \emph{peripherals} ($d_{R/W}^{P_j}$), the \emph{crossbar} ($d_{R/W}^{R_0}$) and the \emph{bridges} ($d_{R/W}^{Q_j}$) from Section~\ref{sec:wcea}.}
We \review{connect the generic controller $CG_i$ to the IP under analysis for these experiments. We let $CG_i$} issue 100'000 transactions, \review{one at a time}, with random burst length ($\beta_i$). We monitor the service times and then pick the longest ones \review{for different $\beta_i$}.

\review{Figure~\ref{fig:experiments} compares the maximum measured experimental delays with the upper bound proposed in Section~\ref{sec:wcea}.}
Figure~\ref{fig:experiments}(a) reports the maximum service time of the main memory subsystem in case of a miss as a function of the burst length of the transaction under analysis, either when (i) only a refill is necessary and (ii) both refill and eviction are necessary, compared with the bounds proposed in Section~\ref{analysis:ms}. 
The measured service times are lower than the bounds. The pessimism is between 3\% and 10.1\%; the larger $\beta_i$, the higher the pessimism.
Higher pessimism on longer transactions is due to the internal splitting at the LLC.
As from our analysis, the memory subsystem is not fully pipelined ($\rho^{MS}=0$).
However, in practice, the control and data phases of consecutive sub-transactions might be partially served in parallel by the LLC and the memory controller.
This means that the longer the transaction, the higher the number of sub-transactions and their overlap, and the lower the service time compared to our model. 
Thus, the pessimism increases.
Figure~\ref{fig:experiments}(b) reports the measured results on the main memory subsystem but in case of a hit, compared with the bounds proposed in Section~\ref{analysis:ms}.
As we consider an LLC hit, the access to the HyperRAM is not performed: this test analyzes the service time of the LLC.
Our bounds are always upper bounds for the maximum measured results. 
The trend here is reversed w.r.t. Figure~\ref{fig:experiments}(a) -- as $\beta_i$ increases, the relative pessimism decreases from 7.7\% down to 0.4\%.
In this case, the source of the pessimism comes only from the control time, which does not depend on $\beta_i$, while there is no pessimism on the data time. 
Hence, this pessimism gets amortized as the burst length and the overall service time increase.
We conduct the same experimental campaign also on the AXI SPM -- the measured results, compared with the bounds proposed in Section~\ref{ssec:l2spm}, are reported in Figure~\ref{fig:experiments}(c).
The trends are similar to the ones reported in Figure~\ref{fig:experiments}(b) for LLC hits -- the pessimism of our analysis is limited to 1 and 2 clock cycles for reads and writes on the control time, respectively.
As in the case of the LLC HITs, the upper bound on the control overhead gets amortized for longer transactions, and the pessimism reduces from 8.8\% to 0.5\%.

Figure \ref{fig:experiments}(d) reports the maximum measured latency to cross an AXI CDC FIFO as a function of the manager clock period (the subordinate clock period is fixed to 30 ns) and compared with the bounds proposed in Section~\ref{ssec:axi_cdc_fifo}. The results are independent of the length of the transaction. 
To stimulate the highest variability, the phases of the clocks are randomly selected on a uniform distribution.
The first bar reports the crossing delays from the manager to the subordinate side, corresponding to the delays introduced on the AW, W, and AR AXI channels.
The second bar reports the crossing delays from the subordinate to the manager side, corresponding to the overall delays on the AXI R and B channels.
The third bar shows the overall delay on a complete transaction, corresponding to the sum of the two previously introduced contributions (see Section~\ref{ssec:axi_cdc_fifo}). 
The pessimism of our bounds is, at most, one clock cycle of the slowest clock between manager and subordinate. 

Figure~\ref{fig:experiments}(e) reports the measured propagation delays introduced by the crossbar over an entire write and read transaction, compared with the bounds of Section~\ref{analysis:xbar}, varying the number of \emph{controllers}. 
As explained in Section \ref{analysis:xbar}, the propagation delay is the sum of the propagation latency without interference (eq. \ref{eq:xbar_prop}) and the additional contention latency (eq. \ref{eq:xbar_contention}), which depends on the number of \emph{controllers}. 
Thanks to the simplicity of the arbitration operated by the crossbar (pure round-robin), our proposed bounds exactly match the measurements. 
We conducted the experimental campaign also on the IO subsystem.
We measured the maximum service time and compared it with the upper bounds of Section~\ref{analysis:io}, which we do not show for space reasons: such IP supports only single-word transactions. 
Our upper bounds exceed the maximum measured service time with pessimism of down to 2 clock cycles, with service times of 4 (write) and 5 (read) clock cycles.

\subsubsection{Parallelism}
\review{We also demonstrate our analysis of parallelism of the \emph{peripherals} ($\chi_{R/W}^{P_j}$) and the \emph{crossbar} ($\chi_{R/W}^{R_0}$) analyzed in Section~\ref{sec:wcea}.}
\review{To do so,} we configured $CG_i$ to issue unlimited outstanding transactions to the \emph{peripheral} under test. In parallel, we monitor the maximum number of accepted outstanding transactions.
Our measurements match our analysis: the maximum number of outstanding transactions is defined by the maximum parallelism accepted at the input stage of the peripherals and the crossbar.

\subsection{System-level experiments}\label{ss:exper-system}

While the previous experiments focused on the evaluation at the IP level, this set of experiments aims to evaluate the system-level bounds proposed in Section~\ref{sec:system-level-analysis}.
We first validate our analysis in simulation.
\review{We developed a System Verilog testbench with two configurable AXI synthetic \review{\emph{controllers}} $CG_i$ connected to the target architecture (see Figure~\ref{fig:he-arch}) stimulating overload conditions to highlight worst-case scenarios.}
We also validate our results on FPGA, generating traffic with CVA6 and the PULP cluster.

At first, we evaluate the results in isolation \emph{at the system level} as a function of the burst length, leveraging the same strategy used for the previous experiments.
\review{Namely, these tests are meant to validate Lemma \ref{lemma:st_single_iso} (eq. \ref{eq:one_trx_iso}).}
\review{To measure the service time at the system level in isolation, we let one $GC_i$ issue 100'000 transactions, one at a time, with different $\beta_i$, while the other $GC_k$ is inactive. We monitor the service times and then pick the longest ones for each $\beta_i$.}
Figures~\ref{fig:sl-experiments} (a) and (b) report the maximum measured system-level response times in isolation for completing a transaction issued by the generic \emph{controller} $GC_i$ and directed to (a) the main memory subsystem (case of cache miss, causing either refill or both refill and eviction) and (b) to the SPM memory, compared with the bounds proposed in Lemma~\ref{lemma:st_single_iso}. 
The measured service times are smaller than the bounds in all the tested scenarios.
The measure and the trends reported in Figure~\ref{fig:sl-experiments}(a) are aligned with the ones found at the IP level and reported in Figure~\ref{fig:experiments}(a). 
This is because the overhead introduced by the crossbar (in isolation) and the CDC FIFOs is negligible compared to the memory subsystem's service time. 
Figure~\ref{fig:sl-experiments}(b) shows a trend aligned with the results at the IP-level reported in Figure~\ref{fig:experiments}(c): the lower $\beta_i$, the higher the pessimism. 
It is worth mentioning that the analysis shows higher pessimism at the system level than at the IP level.
This is due to the extra pessimism from the crossbar and the CDC, which is nevertheless amortized on longer transactions, down to 1.9\%.

We now analyze the results under maximum interference, \review{to verify the results of Lemma~\ref{lemma:max_num_intf_read} and \ref{lemma:intf_diff_kind} and Theorem \ref{theorem:max_intf}}. 
\review{For these tests, the execution of $GC_i$ (100'000 transactions, one at a time) receives interference by \emph{controller} $GC_k$. $\beta_k$ is fixed and equal to $\beta_i$, while we vary the amount of outstanding transactions $GC_k$ can issue ($\phi_{R/W}^{CG_k}$).}
Figures~\ref{fig:sl-experiments} (c), (d), and (e) report the maximum measured system-level response times for completing a transaction issued by the generic \emph{controller} $GC_i$ and directed to (c) the main memory with an LLC miss considering $\beta_i=16$, (d) the SPM memory, considering $\beta_i=16$, and (e) the SPM memory, considering $\beta_i=256$, and compare them with the upper bounds proposed in equation~\ref{eq:theorem}. 
Figures~\ref{fig:sl-experiments} (c), (d), and (e) verify the results of Lemma~\ref{lemma:max_num_intf_read} and \ref{lemma:intf_diff_kind}: when $\phi_{R/W}^{CG_k} > \chi_{R/W}^{MS}$ (two bars on the right), the total service time is defined by the parallelism of the peripheral itself -- as expected, after saturating the number of interfering transactions accepted by the peripheral, the measured results are the same regardless of the increase of $\phi_{R/W}^{CG_k}$. 
Differently, when $\phi_{R/W}^{CG_k} \leq \chi_{R/W}^{MS}$, a reduced value of $\phi_{R/W}^{CG_k}$ corresponds to lower interference and response times.
Figure~\ref{fig:sl-experiments}(c) refers to the case of an LLC miss under interference when $\beta_k=16$.
The results confirm the safeness of our analysis, which correctly upper bounds the overall response times with a pessimism around $15\%$, which is slightly higher than the pessimism of a transaction in isolation at the system level. 
As explained in the previous subsection, when multiple transactions are enqueued, the memory subsystem can partially serve their data and control phases in parallel. 
However, our model only allows $\rho^{MS}=1$ or $\rho^{MS}=0$, i.e., either the \emph{peripheral} is fully pipelined or not pipelined at all.
Since $\rho^{MS}=0$, the pessimism is slightly higher when more transactions are enqueued (and partially served in parallel) as equation \ref{eq:max_intf_1trx} counts the service time of a transaction fully when $\rho^{MS}=0$. 
Varying $\beta_k$ of $GC_k$ gives comparable results -- we do not report such results for briefness and lack of space. 
We provide two charts for the SPM, in Figure~\ref{fig:sl-experiments}(d) and Figure \ref{fig:sl-experiments}(e). 
The comparison of the two charts highlights how the interfering transactions' length impacts the analysis's pessimism, ranging between $19.7\%$ for $\beta=16$ to $1\%$ for $\beta=256$.
The trend here is aligned with the service time at the system level in isolation: the pessimism comes from the control times of SPM and propagation latency of the crossbar and the CDC, which are amortized as the data time increases with $\beta_k$.

\begin{figure*}[t]
    \centering
    \includegraphics[width=\textwidth]{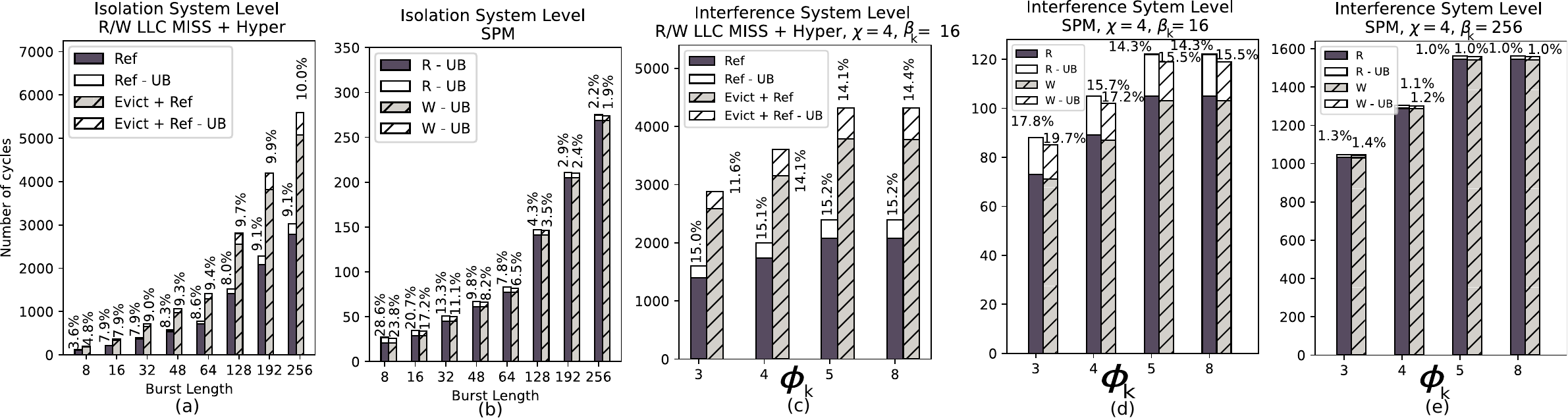}
    \caption{Services times under interference.}
    \label{fig:sl-experiments}
\end{figure*}

\subsection{Discussion}
\review{In this Section, we validated the analysis of Sections ~\ref{sec:wcea} and \ref{sec:system-level-analysis} through an extensive set of tests.}
\review{We demonstrated how the proposed approach enables detailed explanations of the analysis's pessimism and facilitates iterative refinement.}
\review{This allows us to derive upper bounds that are safe yet not overly pessimistic, particularly when compared to similar state-of-the-art works based on closed-source or loosely-timed IPs.}
\review{Nevertheless, while the methodology is promising, the resulting analysis may seem limited in comparison to other works that model more sophisticated closed-source IPs.}
\review{Here, we discuss the limitations of our analysis, focusing on its dependence on the underlying characteristics of the available open-source hardware.}

\review{It is noteworthy how the analysis leverages the round-robin policy of the main interconnect and the in-order nature of \emph{peripherals} in Lemmas \ref{lemma:max_num_intf_read} and \ref{lemma:intf_diff_kind}.}
\review{The absence of internal reordering allows to derive the number of transactions preceding the one under interference directly from the arbitration policy.}
\review{As long as the \emph{peripherals} serve the transactions in order, extending the analysis to support other arbitration policies is expected to require minimal effort.}
\review{Instead, supporting \emph{peripherals} with internal transaction reordering can lead to \emph{timing anomalies}\cite{enabling_compositionability} and make the proposed model unsafe, as previously demonstrated in \cite{ooocoresareSOnasty}.}
\review{Our analysis focuses on the available \emph{peripherals} within the target architecture, as out-of-order \emph{peripherals} are not available open-source to us.}
\review{We envision expanding the analysis to match higher-performance platforms as open-source hardware evolves.}

\review{Lastly, it is important to note that the analysis bounds only a single transaction issued by $C_i$ -- this limitation is not imposed on the interfering controllers.}
\review{Lemma \ref{lemma:max_num_intf_read} does not consider $C_i$ to have more pending transactions, except for the ones already accepted by $P_j$.}
\review{In other words, Lemma \ref{lemma:max_num_intf_read} assumes that there is not a queue of transactions buffered in the \emph{bridges} between $C_i$ and $R_0$, which could exist when $P_j$ is full.}
\review{We could potentially extend the model to define a batch of enqueued transactions and then modify Lemma \ref{lemma:max_num_intf_read} to analyze this scenario.}
\review{Such an extension would further build upon the proposed model and analysis, which is limited to bound the access time of a single transaction.}

\section{Related Work}\label{sec:related}

\begin{table*}[ht]
\scriptsize
\centering
\caption{Comparison with State-of-the-Art works for predictability. IC = Interconnect. DMR = Deadline miss ratio.}
\label{tab:soa}
\begin{tabular}{|c|c|c|c|c|c|}
\hline
& \multirow{1}{*}{Target}     
%& \multirow{1}{*}{Model}  
& \multirow{1}{*}{\begin{tabular}[c]{@{}c@{}}Analysis  on\end{tabular}}
& \multirow{1}{*}{\begin{tabular}[c]{@{}c@{}}Pessimism\end{tabular}}
& \multirow{1}{*}{\begin{tabular}[c]{@{}c@{}}Technology\end{tabular}} 
& \multirow{1}{*}{\begin{tabular}[c]{@{}c@{}}Open  RTL\end{tabular}} 
\\ \hline

\multirow{1}{*}{\begin{tabular}[c]{@{}c@{}}Biondi  et. al. \cite{MPAM-analysis}\end{tabular}}
& \multirow{1}{*}{\begin{tabular}[c]{@{}c@{}}ARM MPAM  Protocol\end{tabular}}
%& \multirow{1}{*}{\cmark}   
& \multirow{1}{*}{\begin{tabular}[c]{@{}c@{}}Protocol specs (Model)\end{tabular}}                           
& \multirow{1}{*}{No HW}   
& \multirow{1}{*}{\xmark}
& \multirow{1}{*}{\xmark}  \\ \hline

%\multirow{1}{*}{\begin{tabular}[c]{@{}c@{}}MCSIM  \cite{mirosanlou2020mcsim}\end{tabular}}
%& \multirow{1}{*}{\begin{tabular}[c]{@{}c@{}}DDR3  CTRL\end{tabular}}
%& \multirow{1}{*}{\cmark}                           
%& \multirow{1}{*}{\begin{tabular}[c]{@{}c@{}}IP  specs\end{tabular}}                           
%& \multirow{1}{*}{\xmark}   
%& \multirow{1}{*}{\xmark}
%& \multirow{1}{*}{\xmark}  \\ \hline

\multirow{1}{*}{\begin{tabular}[c]{@{}c@{}}Hassan et. al. \cite{hassanzoni}\end{tabular}}
& \multirow{1}{*}{\begin{tabular}[c]{@{}c@{}}JEDEC DDR3 Protocol\end{tabular}}
%& \multirow{1}{*}{\cmark}                           
& \multirow{1}{*}{\begin{tabular}[c]{@{}c@{}}Protocol specs (Model)\end{tabular}}                           
& \multirow{1}{*}{$0\% - 200\%$}
& \multirow{1}{*}{\xmark}
& \multirow{1}{*}{\xmark}  \\ \hline

\multirow{1}{*}{\begin{tabular}[c]{@{}c@{}}\review{Abdelhalim et.al.}  \cite{ooocoresareSOnasty}\end{tabular}}
& \multirow{1}{*}{\begin{tabular}[c]{@{}c@{}}Whole mem. hier.\end{tabular}}
& \multirow{1}{*}{\begin{tabular}[c]{@{}c@{}}IPs \& System (C++ Model)\end{tabular}}                           
& \multirow{1}{*}{$16\% - 50\%$}                           
& \multirow{1}{*}{\xmark}
& \multirow{1}{*}{\xmark}  \\  \hline

\multirow{1}{*}{\begin{tabular}[c]{@{}c@{}}BlueScale  \cite{jiang2022bluescale}\end{tabular}}
& \multirow{1}{*}{\begin{tabular}[c]{@{}c@{}}Hier.  mem. IC\end{tabular}}
%& \multirow{1}{*}{\cmark}
& \multirow{1}{*}{\begin{tabular}[c]{@{}c@{}}IC  uArch (Black-box)\end{tabular}}                           
& \multirow{1}{*}{DMR}                           
& \multirow{1}{*}{FPGA}
& \multirow{1}{*}{\xmark}  \\  \hline

\multirow{1}{*}{\begin{tabular}[c]{@{}c@{}}AXI-RT-IC  \cite{jiang2022axi} \end{tabular}}
& \multirow{1}{*}{\begin{tabular}[c]{@{}c@{}}AXI SoC  IC\end{tabular}}
%& \multirow{1}{*}{\cmark}
& \multirow{1}{*}{\begin{tabular}[c]{@{}c@{}}IC  uArch (Black-box)\end{tabular}}                           
& \multirow{1}{*}{DMR}                        
& \multirow{1}{*}{FPGA}
& \multirow{1}{*}{\xmark}  \\ \hline

%\multirow{1}{*}{\begin{tabular}[c]{@{}c@{}}\review{HyperConnect}  \cite{restuccia2020axi}\end{tabular}}
%& \multirow{1}{*}{\begin{tabular}[c]{@{}c@{}}AXI Hier.  mem. IC\end{tabular}}
%& \multirow{1}{*}{\xmark}
%& \multirow{1}{*}{\xmark}                           
%& \multirow{1}{*}{\xmark}                           
%& \multirow{1}{*}{FPGA}
%& \multirow{1}{*}{\xmark}  \\  \hline

\multirow{1}{*}{\begin{tabular}[c]{@{}c@{}}Restuccia et. al. \cite{restuccia2022bounding}\end{tabular}}
& \multirow{1}{*}{\begin{tabular}[c]{@{}c@{}}AXI Hier.  mem. IC\end{tabular}}
%& \multirow{1}{*}{\cmark}
& \multirow{1}{*}{\begin{tabular}[c]{@{}c@{}}IC  uArch (Black-box)\end{tabular}}                           
& \multirow{1}{*}{$50\%-90\%$}   
& \multirow{1}{*}{FPGA}
& \multirow{1}{*}{\xmark}  \\  \hline

\multirow{1}{*}{\begin{tabular}[c]{@{}c@{}}\review{AXI-REALM}  \cite{AXI-REALM}\end{tabular}}
& \multirow{1}{*}{\begin{tabular}[c]{@{}c@{}}AXI traffic  regulator\end{tabular}}
& \multirow{1}{*}{No analysis}                           
& \multirow{1}{*}{No model}                           
& \multirow{1}{*}{FPGA \& ASIC}
& \multirow{1}{*}{\cmark}   \\ \hline

\multirow{1}{*}{\begin{tabular}[c]{@{}c@{}}Ditty  \cite{wu2023ditty}\end{tabular}}
& \multirow{1}{*}{\begin{tabular}[c]{@{}c@{}}Cache coher.  mechanism\end{tabular}}
%& \multirow{1}{*}{\cmark}
& \multirow{1}{*}{\begin{tabular}[c]{@{}c@{}}IP (Fine-grained RTL)\end{tabular}}                           
& \multirow{1}{*}{$100\%-200\%$}   
& \multirow{1}{*}{FPGA}
& \multirow{1}{*}{\cmark} \\ \hline

\multirow{1}{*}{\begin{tabular}[c]{@{}c@{}}This  Work\end{tabular}}
& \multirow{1}{*}{\begin{tabular}[c]{@{}c@{}}SoC IC,  peripherals \&  system-level\end{tabular}}
%& \multirow{ 1}{*}{\cmark}
& \multirow{1}{*}{\begin{tabular}[c]{@{}c@{}}IP \& System (Fine-grained RTL)\end{tabular}}                         
& \multirow{1}{*}{$1\%-28\%$}   
& \multirow{1}{*}{\begin{tabular}[c]{@{}c@{}}FPGA \&  ASIC\end{tabular}}
& \multirow{1}{*}{\cmark} \\  \hline
\end{tabular}
\end{table*}

\review{In this Section, we provide a thorough comparison with previous works focusing on enhancing the timing predictability of digital circuits.}
\review{Traditionally,} the majority of these works leverage commercial off-the-shelf devices\cite{survey_multi_core_scs,thewcetproblem} or predictable architectures modeled with a mix of cycle-accurate and behavioral simulators\cite{TCREST}. 
Also, they focus on bounding the execution times for predefined specific software tasks rather than the individual transaction service times\cite{survey_multi_core_scs,confidence,TCREST,enabling_compositionability}. 
\review{Furthermore, they} build the models from dynamic experiments rather than from static analysis, largely due to the dearth of detailed hardware specifications\cite{mitraDandT}, \review{limiting the generality of their approach.}
\review{More recent works advocate for static modeling and analysis of protocols\cite{MPAM-analysis,hassanzoni}, interconnect\cite{jiang2022bluescale,restuccia2022bounding,jiang2022axi}, and shared memory resources\cite{ooocoresareSOnasty,wu2023ditty} to provide more generic and comprehensive models.}
\review{While their value is undeniable, due to the unavailability of the source RTL, each one focuses on only one of these resources, resulting in a significant penalty to the pessimism of the upper bounds\cite{ooocoresareSOnasty}}.
Our work breaks from this convention, presenting a holistic static model of an entire open-source architecture rigorously validated through RTL cycle-accurate simulation and FPGA emulation. 
\review{As Table \ref{tab:soa} shows, this is the first work to analyze and model the open-source silicon-proven RTL of all the IPs composing a whole SoC to build the least pessimistic upper bounds for data transfers within the architecture when compared to similar SoA works.}

Biondi et al.~\cite{MPAM-analysis} developed a model of the memory-access regulation mechanisms in the ARM MPAM and provided detailed instantiations of such mechanisms, which they then evaluated with IBM CPLEX, a decision optimization software for solving complex optimization models.
While elegant, this approach is not validated on hardware and, therefore, is limited in terms of applicability and precision.
A more practical and adopted approach is the one proposed by Hassan and Pellizzoni\cite{hassanzoni}.
The authors develop a fine-grained model of the JEDEC DDR3 protocol, validated with MCsim\cite{mirosanlou2020mcsim}, a cycle-accurate C++ memory controller simulator.
Unfortunately, not having access to the RTL prevents fine-grained modeling and analysis and mandates over-provisioning, strongly impacting the overall pessimism of the system, which can be as high as 200\%.
\review{Abdelhalim et al. in \cite{ooocoresareSOnasty} present a study bounding the access times of memory requests traversing the entire memory hierarchy and propose $\mu$architectural modifications to the arbiters in such hierarchy.}
\review{Their modifications result in very low pessimism (down to 16\%) on synthetic and real-world benchmarks. However, the results are validated on C++ models of the cores, interconnect, and memory controllers, not RTL code targeting silicon implementation.}

More recently, different researchers proposed models of hardware IPs that they could validate through cycle-accurate experiments \cite{restuccia2022bounding, jiang2022axi, restuccia2020axi}.
In \cite{restuccia2022bounding}, Restuccia et al. focused on upper bounding the response times of AXI bus transactions on FPGA SoCs through the modeling and analysis of generic hierarchical interconnects arbitrating the accesses of multiple hardware accelerators towards a shared DDR memory.
In this work, the interconnect under analysis is a proprietary Xilinx IP, which had to be treated as a black box.
Also, due to the unavailability of the RTL code, the authors did not model the other IPs composing the target platform, limiting the precision of the proposed upper bounds, which achieve a pessimism between 50\% and 90\%.
Jiang et al. modeled, analyzed, and developed AXI-IC$^{\text{RT}}$\cite{jiang2022axi} and Bluescale\cite{jiang2022bluescale}, two sophisticated interconnects providing predictability features and coming with a comprehensive model.
However, the model and analysis proposed in AXI-IC$^{\text{RT}}$\cite{jiang2022axi}, and Bluescale\cite{jiang2022bluescale} are not as fine-grained as ours: the authors do not provide upper bounds of the access times but rather focus on the deadline miss ratio given a fixed workload for the different controllers in the system. Moreover, the authors do not provide the RTL of such solutions.
\review{AXI-REALM \cite{AXI-REALM} proposes completely open-source IPs supporting predictable communications. However, it misses a holistic model and analysis.}
In Ditty~\cite{wu2023ditty}, researchers propose an open-source predictable directory-based cache coherence mechanism for multicore safety-critical systems that guarantees a worst-case latency (WCL) on data accesses with almost cycle-accurate precision.
However, Ditty's model only covers the coherency protocol latency and the core subsystem, overlooking system-level analysis and achieving very pessimistic boundaries.
In this landscape, it emerges clearly that our work is the first one covering both modeling and analysis of the interconnects and the shared memory resources, with an in-depth analysis of silicon-proven open-source RTL IPs and achieving the lowest pessimism when compared to similar SoA works.

\section{Conclusions}\label{sec:conclusion}

\review{In conclusion, this is the first work to bridge the gap between open-source hardware and predictability modeling and analysis.}
\review{It presented (i) a fine-grained model and analysis for the typical building blocks composing modern heterogeneous low-power SoCs directly based on the source RTL, and (ii) a full mathematical analysis to upper bound data transfer execution times.}
\review{Namely, we demonstrated a methodology that successfully exploits the availability of the source code to provide safe, but not overly pessimistic, upper bounds for the execution times of data transfers when compared to similar SoA works based on closed-source IPs.}

\review{As discussed in Section \ref{sec:exps}, after this thorough evaluation, we envision extending our results to other popular open-source IPs and different arbitration policies.} To hopefully stimulate novel research contributions, we open-source a guide to replicate the results shown in Section \ref{sec:exps} at \url{https://github.com/pulp-platform/soc_model_rt_analysis}, comprehensive of the simulated environment and the software benchmarks to run on a sophisticated Cheshire-based SoC targeting automotive applications.
%\review{Subsequently, we plan to extend our experiments to real-world, end-to-end applications, as discussed in Section \ref{sec:exps}}.

\bibliographystyle{IEEEtran}
\bibliography{IEEEabrv,bib}

\vspace{-11mm}

\begin{IEEEbiographynophoto}
    {Luca Valente} received the MSc degree in electronic engineering from the Politecnico of Turin in 2020. He is currently a PhD student at the University of Bologna in the Department of Electrical, Electronic, and Information Technologies Engineering (DEI). His main research interests are hardware-software co-design of heterogeneous SoCs.
\end{IEEEbiographynophoto}

\begin{IEEEbiographynophoto}{Francesco Restuccia} received a PhD degree in computer engineering (cum laude) from Scuola Superiore Sant’Anna Pisa, in 2021. He is a postdoctoral researcher at the University of California, San Diego. His main research interests include hardware security, on-chip communications, timing analysis for heterogeneous platforms, cyber-physical systems, and time-predictable hardware acceleration of deep neural networks on commercial
FPGA SoC platforms. 
\end{IEEEbiographynophoto}

\begin{IEEEbiographynophoto}{Davide Rossi} received the Ph.D. degree from the University of Bologna in 2012. He has been a Post-Doctoral Researcher with the Department of Electrical, Electronic and Information Engineering “Guglielmo Marconi,” University of Bologna, since 2015, where he is currently an Associate Professor position. His research interests focus on energy-efficient digital architectures. In this field, he has published more than 100 papers in international peer-reviewed conferences and journals. %He is recipient of Donald O. Pederson Best Paper Award 2018, 2020 IEEE TCAS Darlington Best Paper Award, 2020 IEEE TVLSI Prize Paper Award.
\end{IEEEbiographynophoto}

\begin{IEEEbiographynophoto}{Ryan Kastner} is a professor in the Department
of Computer Science and Engineering at UC San Diego. He received a PhD in Computer Science at UCLA, a masters degree in engineering and bachelor degrees in Electrical Engineering and Computer Engineering from Northwestern University. His current research interests fall into three areas: hardware acceleration, hardware security, and remote sensing.
\end{IEEEbiographynophoto} 

\begin{IEEEbiographynophoto}{Luca Benini} holds the chair of Digital Circuits and Systems at ETHZ and is Full Professor at the Universit\`{a} di Bologna. He received a PhD from Stanford University. His research interests are in energy-efficient parallel computing systems, smart sensing micro-systems and machine learning hardware. He has published more than 1000 peer-reviewed papers and 5 books. He is a Fellow of the ACM and a member of the Academia Europaea.    
\end{IEEEbiographynophoto}

\end{document}